\documentclass[12pt, letterpaper]{elsarticle}
\usepackage{tikz}
\usepackage{tikz,fullpage}
\usepackage{amsmath}
\usetikzlibrary{arrows,
topaths}%

\usepackage{amssymb}

\usepackage{multirow}
\usepackage{pgf}
\usetikzlibrary{arrows,automata}
\usepackage{tkz-berge}
\usepackage[position=top]{subfig}
\usepackage{float}
\usepackage[section]{placeins}
\usepackage{arydshln}
\usepackage{latexsym}\usepackage{epic}\usepackage{amsmath,amsthm,amssymb}
\usepackage{graphicx}
\usepackage{url}
\usepackage{pgfplots}
\usetikzlibrary{plotmarks}
\usepackage{subfig}
\usepackage{soul}
\usepackage{afterpage}
\usepackage{bbm}
\usepackage{tikz}\usetikzlibrary{arrows}
\usepackage{cleveref}
\usetikzlibrary{decorations.markings}
\usepackage[left=0.9in,top=0.9in,right=0.9in,bottom=0.9in,nohead]{geometry}
\usepackage{setspace}
\usepackage[norelsize,ruled, lined,linesnumbered]{algorithm2e}
\usepackage{enumitem}
\DeclareMathOperator{\argmax}{argmax}

\newtheorem{theorem}{Theorem}

\newtheorem{lemma}{Lemma}
\newtheorem{proposition}{Proposition}
\newtheorem{definition}{Definition}

\newtheorem{remark}{Remark}
\newtheorem{example}{Example}

\usepackage{natbib}

\usepackage{xcolor}
\usepackage[colorinlistoftodos,prependcaption,textsize=tiny]{todonotes}
\usepackage{xargs}
\newcommandx{\OlegComment}[2][1=]{\todo[linecolor=cyan,backgroundcolor=cyan!25,#1]{#2}}

\usepackage[symbol]{footmisc}

\begin{document}
\onehalfspace
\begin{frontmatter}




\title{On Greedy and Strategic Evaders in Sequential Interdiction Settings\\ with Incomplete Information}


\author[label1]{Sergey S.~Ketkov}
\author[label1,label2]{Oleg A.~Prokopyev\footnote[2]{Corresponding author. Email: droleg@pitt.edu; phone: +1-412-624-9833.}}
\address[label1]{Laboratory of Algorithms and Technologies for Networks Analysis, National Research University \\Higher School of Economics (HSE), Bolshaya Pecherskaya st., 25/12, Nizhny Novgorod, 603155, Russia}
\address[label2]{Industrial Engineering, University of Pittsburgh, 3700 O'Hara Street, 1048 Benedum Hall, Pittsburgh, PA 15261, USA}
\vspace{-2mm}

\begin{abstract}

\looseness-1 We consider a class of sequential network interdiction problem settings where the interdictor has incomplete initial information about the network while the evader has complete knowledge of the network including its structure and arc costs. In each decision epoch, the interdictor can block (for the duration of the epoch) at most $k$ arcs known to him/her. By observing the evader's actions, the interdictor learns about the network structure and costs and thus, can adjust his/her actions in subsequent decision epochs. It is known from the literature that if the evader is greedy (i.e., the shortest available path is used in each decision epoch), then under some assumptions 
the greedy interdiction policies that block $k$-most vital arcs in each epoch are efficient and have a finite regret. In this paper, we consider the evader's perspective and explore deterministic ``strategic'' evasion policies under the assumption that the interdictor is greedy. We first study the theoretical computational complexity of the evader's problem. Then we derive basic constructive properties of optimal evasion policies for two decision epochs when the interdictor has no initial information about the network structure. These properties are then exploited for the design of a heuristic algorithm for a strategic evader in a general setting with an arbitrary time horizon and any initial information available to the interdictor. Our computational experiments demonstrate that the proposed heuristic outperforms the greedy evasion policy on several classes of synthetic network instances under either perfect or noisy information feedback. Finally, some interesting insights from our theoretical and computational results conclude the paper. 
\end{abstract}


\end{frontmatter}

\section{Introduction}\label{sec:intro}

\looseness-1 Network interdiction is typically posed as a one- or multi-stage decision-making problem in a network with two decision-makers, an \emph{interdictor} and an \emph{evader}. The evader traverses in the network (e.g., between two fixed nodes, a source and a destination), while the interdictor aims to disrupt to the maximum possible extent (or completely stop) the evader's movement. Perhaps, the most well-known and studied network interdiction problem is the \emph{shortest path interdiction problem}~\cite{5}, where the interdictor seeks a set of arcs whose removal subject to some budgetary constraint, maximizes the cost of the shortest path between two specified nodes. That is, the evader is assumed to move along the shortest path between these nodes. If the interdictor is allowed to block at most $k$ arcs, then this problem is often referred to as the $k$-\emph{most vital arcs problem} \cite{11}. The shortest path interdiction problem is known to be $NP$-hard~\cite{11}. Nevertheless, rather effective exact solution approaches for various versions of the network interdiction problem exist in the literature, see, e.g., \cite{5,sullivan2014exact}.

\looseness-1Network interdiction forms a broad class of deterministic and stochastic optimization problems with applications mostly arising in the military, law-enforcement and infectious disease control contexts, see surveys in~\cite{Dimitrov2013,smith2008algorithms,smith2013modern,Smith2019,wood2011bilevel} and the references therein. Perhaps, the earliest example of the network interdiction considerations can be found in a now de-classified RAND report for the U.S. Air Force in 1955~\cite{harris1955fundamentals}, which studied the Soviet railway network. There has been a substantial increase in the interest for this class of optimization problems since the early 2000's given concerns with various homeland security issues. Recent examples of real-world settings, where different types of interdiction models have been applied include nuclear smuggling interdiction~\cite{allain2016evolving,morton2007models}, drug trafficking enforcement~\cite{enayaty2018logic,malaviya2012multi}, border patrol~\cite{brown2006defending}, the problem of interdicting a nuclear-weapon project \cite{brown2009interdicting}, etc. Furthermore, interdiction models have been applied for analysis of critical infrastructure systems, see examples of real-life case studies in~\cite{brown2006defending} for the U.S. strategic petroleum reserve and an electrical transmission grid. A brief overview of the network interdiction settings in the modern military contexts can also be found in~\cite{allain2016evolving}.

In the typical shortest path interdiction problem the graph represents an underlying ``infiltration'' or ``smuggling'' network (with a finite number of possible ``infiltration'' or ``smuggling'' routes, respectively) in which the evader travels or moves the illegal materials. The arc costs may correspond to appropriately defined detection (or non-detection) probabilities or the amount of the evader's effort required to smuggle/move a unit of illegal materials through an arc, while graph nodes often represent some geographical locations. The interdiction actions correspond to placing sensors, road/helicopter patrols or actual attacks (in the military contexts) to prevent successful evasions, see discussions in~\cite{allain2016evolving,smith2008algorithms, smith2013modern, wood2011bilevel}. Finally, we also refer the reader to~\cite{brown2006defending,Dimitrov2013}, where the authors provide interesting discussions on some practical implications and insights from various interdiction models in the literature.

While most of the studies in the network interdiction literature consider deterministic settings, a number of more recent works consider the network interdiction problem in stochastic settings;  see, e.g., \cite{3, Pan2008, Janjarassuk2008} and the survey in \cite{Nehme2009}. Typically such models assume that either the outcomes of interdiction actions are uncertain or there is uncertainty with respect to the evader's actions. Then the interdictor's objective is to optimize some utility function over the network, e.g., minimize the expected cost of the maximum-reliability path \cite{Pan2008}.

In addition, several studies consider multi-stage stochastic interdiction models. One example is the study by Held~et~al.~\cite{Held2005}, who assumes that the network's configuration itself is subject to uncertainty. The interdictor attempts to maximize the probability that the minimum path cost exceeds a given threshold. Since even computing an objective function value is time-consuming some heuristic algorithms are proposed and their effectiveness is validated numerically.

Furthermore, an interesting dynamic deterministic version of the network interdiction problem has been recently considered in \cite{sefair2016dynamic}, where the evader can dynamically adjust her\footnote{Note that in the remainder of the paper we refer to the interdictor and the evader as ``he/his'' and ``she/her,'' respectively.} movement at every node of her path in the network by observing the interdictor's actions. The interdictor, in turn, can interdict arcs any time the evader arrives at a node in the network. It is also assumed that the interdictor has a limited interdiction budget.

\looseness-1The current study is motivated and builds upon recent works of Borrero et al.~in \cite{1,2}. Specifically, {in their network model in~\cite{1}} the interdictor and the evader interact sequentially over multiple {decision epochs (or rounds)\footnote{We use the terms ``decision epoch'' or ``round'' interchangeably}}. In each {epoch}, the interdictor can block at most $k$ arcs for the duration of the current decision {epoch}, while the evader is assumed to be greedy, i.e., in each {epoch} the evader traverses along the shortest path between two fixed nodes in the interdicted network.
The evader's loss in each epoch (i.e., her instantaneous loss) is equal to the cost of the shortest path in the interdicted network. This modeling approach can be justified, for example, by interdiction and evasion dynamics arising in monitoring and patrolling problems, where the interdictor has to periodically reallocate his resources over different geographical locations, see additional discussion in~\cite{1}.

\looseness-1 The key feature of the model in \cite{1} is that the interdictor has incomplete initial information about the network including its structure and costs, but learns about the network structure and arc costs by observing the evader's actions (i.e., the evasion path) in each decision epoch. The learning component is motivated by practical settings, where the interdictor can observe the evader's actions (for example, by using a satellite or a drone), but cannot immediately react upon those actions; see, e.g., \cite{Zheng2012}. In particular, it is assumed that this information feedback is deterministic and perfect that is, the interdictor learns about the existence and the exact costs of the arcs used by the evader in the previous decision epochs. The quality of an interdiction policy is measured using either cumulative regret or time stability. The former is defined as the difference in the total cost (over some predefined finite number of decision epochs) incurred by the evader under the current interdiction policy against the policy of an oracle interdictor with prior complete knowledge of the network. Clearly, the oracle interdictor implements an optimal solution of the $k$-most vital arcs problem in each decision epoch. Time stability is defined as the number of rounds that are necessary for the policy to gain sufficient amount of information in order to implement a solution of the $k$-most vital arcs problem that is also optimal in the full information network for the remainder of the time horizon. 

\looseness-1 The main results of \cite{1} can be summarized as follows. First, it is shown that for their deterministic setting, in general, there do not exist policies that perform better than any other policy for any graph consistent with the initial information available to the interdictor. Thus, the focus of \cite{1} is on the \emph{greedy and robust interdiction policies}. They are greedy because they block a set of the $k$-most vital arcs from the network known to the interdictor in each round; they are robust because whenever the exact cost of the arc is not known to the interdictor, then the policies assume the worst-case scenario for the evader. These policies turn out to be ``efficient'' in the following sense: (\emph{i}) they eventually find and maintain an optimal solution to the $k$-most vital arcs problem in the full information network (i.e., an optimal solution of the oracle interdictor), within a finite number of decision epochs (possibly, instance dependent); and (\emph{ii}) this class of policies is \emph{not dominated} that is, for any possible instance of the initial information available to the interdictor and any policy that is not greedy and robust, there exists a greedy and robust policy that is strictly better (with respect to either the cumulative regret or time stability) than the aforementioned non-greedy and/or non-robust policy for some graph that is consistent with the initial information available to the interdictor.

Property (\emph{i}) of greedy and robust policies also implies that they have a finite regret. Furthermore, these policies detect when the instantaneous regret reaches zero in real time, i.e., when an optimal solution of the oracle interdictor is found. In addition to these attractive theoretical properties, the results of computational experiments~in \cite{1} also confirm the superiority of greedy and robust interdiction policies against several other benchmark policies. Finally, in \cite{2} the authors generalize the theoretical results and greedy policies from~\cite{1} for a more general class of max-min linear mixed-integer problems.

Note that in \cite{1,2}, similar to the vast majority of the related interdiction literature, the authors focus on the interdictor's perspective. However, given that the outlined greedy and robust interdiction policies are {rather intuitive} and very simple to implement (which is important for their applicability in practical applications) and have attractive theoretical properties, it is natural to explore the evader's perspective. In particular, if we assume that the interdictor is greedy and robust, what are good strategic policies for the evader against greedy and robust interdiction policies? Do such policies have any interesting structural properties? How can they be constructed? These research questions form the main focus of the current study.

Admittedly, our modeling approach is somewhat stylized and clearly not without limitations. However, our goal is to gain insights about the structure and properties of the evasion policies that are effective against greedy and robust interdictors. We envision that the results of the study can be further exploited for the development of more advanced interdiction models and policies, where the evader follows more sophisticated behavior than that typically assumed in the standard interdiction models.

\looseness-1 In this paper, we consider the deterministic setting similar to \cite{1} with an additional modification. Specifically, with respect to the initial information available to the interdictor, in \cite{1} it is assumed that for some arcs the interdictor is aware only about some upper and lower bounds on their costs. In contrast, we assume that whenever the existence of the arc is known to the interdictor initially, then its cost is also known, which in fact, is a less favorable scenario for the evader, whose perspective we consider. Furthermore, recall that in the considered setting the feedback is perfect that is, whenever the arcs is used by the evader, its existence and precise cost information is revealed to the interdictor. Therefore, under our rather mild assumption on the initial information available to the interdictor, the greedy and robust interdiction policies of \cite{1,2} reduce to simply the greedy ones.

 As mentioned earlier, the evader's perspective is discussed by relatively few studies. In particular, Sanjab~et~al. \cite{Sanjab2017, Sanjab2019} study a one-stage zero-sum game, where the evader transports products between two specified nodes in the network and attempts to minimize the expected delivery time. Furthermore, the authors focus on a bounded rationality model, that is, the risk level of a selected path or the merit of a selected attack location can be perceived subjectively by the decision-makers. However, in our setting the bounded rationality of the evader is addressed by considering strategic decisions in a multi-period deterministic setting. Additionally, in contrast to \cite{Sanjab2017, Sanjab2019} we assume limited information about the network's structure, but not uncertainty in the evader's (or interdictor's) decisions.

Eventually, our formulation of the evader's problem can be viewed as a particular case of online combinatorial optimization problem with sleeping experts \cite{14}, which is a generalization of the classical multi-armed bandit formulation; see, e.g., \cite{16}. The problem with stochastic loss functions and adversarial availability of actions is discussed in \cite{20, 15}. The evader's problem is related to the online adversarial shortest path problem \cite{20}. Nevertheless, uniform mixing assumption \cite{20} does not hold for deterministic strategies and the notion of regret (for example, per action regret \cite{14, 17}) compares an arbitrary policy with a policy of using the best ranked action in the hindsight, i.e., with the greedy evasion policy. A distinctive feature of our setting is that the evader's actions determine the information collected by the interdictor and, thus, influence the structure of the setting.

In view of the discussion above, the contribution of this paper (and its remaining structure) can be summarized as follows:
\begin{itemize}
 \item In Section~\ref{sec:mathmod}, we formulate the  repeated  evasion model with multiple decision epochs under the assumption that the interdictor follows a greedy interdiction policy. In the deterministic setting with perfect information feedback, this model can be viewed as a particular class of hierarchical deterministic combinatorial optimization problems, which is the main focus of this study.
 \item \looseness-1 In Section~\ref{sec:NP}, we explore the theoretical computational complexity of the considered evader's problem. We show that the evader's problem is $NP$-hard in the case of two decision epochs as long as there are no restrictions for the initial information available to the interdictor. This result is established for networks where the interdiction problem is polynomially solvable.
 \item In Section~\ref{sec:theory}, we provide theoretical analysis of evasion policies in the setting with two decision epochs, where the interdictor has no initial information about the network arcs. We show that under some mild assumption, the optimal evasion decision is either greedy, or consists of two distinct paths that intersect (i.e., have some arcs in common) with the overall shortest path in the network.
 \item In Section~\ref{sec:algorithm}, we exploit these theoretical properties to develop a heuristic algorithm for the strategic evader in a more general setting with an arbitrary time horizon and no restrictions on the initial information available to the interdictor.

 \item In Section~\ref{sec:experiment}, we perform computational experiments that demonstrate that the proposed heuristic consistently outperforms the greedy evasion policy on several classes of synthetic network instances. In our experiments, in addition to perfect information feedback we also consider feedback scenarios, where the information obtained by the interdictor from the evader's actions is noisy.
\end{itemize}

\looseness-1 Finally, Section~\ref{sec:conclusions} concludes the paper, {highlights interesting insights from our theoretical results and computational observations}, and then outlines promising directions for future research.

\noindent \textbf{Notation.} In the remainder of the paper we use the following notation. Let $G = (N, A)$ be a connected weighted directed graph, where $N$ and $A$ denote its sets of nodes and directed arcs, respectively. Denote by $c_a$ a nonnegative arc cost associated with each arc $a \in A$. For $A' \subseteq A$ we define a subgraph of $G$ induced by this subset of arcs as $G[A'] = (N, A')$. Denote by $T\in\mathbb{Z}_{>0}$ the number of {decision epochs (or rounds)} and by $k \in \mathbb{Z}_{>0}$ the interdictor's budget.

Furthermore, we define two particular nodes in $G$, which are referred to as the source, $s$, and destination, $f$, nodes, respectively. Let $\mathcal{P}_{sf}(G)$ be a set of all simple directed paths from $s$ to $f$ in network~$G$. Any path $P \in \mathcal{P}_{sf}(G)$ is given by a sequence of arcs $(s,v_1),(v_1,v_2),\ldots,(v_{|P|-1},f)$, which we denote by $\{s \rightarrow v_1 \rightarrow \ldots \rightarrow v_{|P|-1} \rightarrow f\}$ for convenience. Also, let $\ell(P)$ be the cost of path $P \in \mathcal{P}_{sf}(G)$, that is $\ell(P)= \sum_{a \in P} c_a.$ Finally, define $$z(G)=\min_{P\in \mathcal{P}_{sf}(G)}\ell(P),$$ \noindent i.e., $z(G)$ is the cost of the shortest path from $s$ to $f$ in $G$.

\section{Mathematical Model}\label{sec:mathmod}

We consider a sequential decision-making process, where in each decision epoch (round) $t\in \{1,2,\ldots,T\}$ an evader and an interdictor interact. The evader has full information about the underlying network, while the interdictor has limited information about its structure and costs. In particular, we assume that the interdictor initially observes a subnetwork $G[A_0]$ of the given network $G = (N, A)$, i.e., he is informed only about the existence of arcs in $A_0$ along with their costs~$\{c_a\}_{a \in A_0}$. Let
\begin{equation}\nonumber
\mathcal{C}_0=(N,A_0),
\end{equation}
where we refer to $\mathcal{C}_0$ as the initial information available to the interdictor as it contains his initial knowledge about the structure and costs of the network.


In each decision epoch $t\in \{1,2,\ldots,T\}$ the following sequence of events takes place:
\begin{enumerate}
\item
The interdictor chooses set $I_{t} \subseteq A_{t-1}$ of at most $k$ arcs to be blocked for the time of exactly one {decision epoch}.

\item
The evader traverses along path $P_t \in \mathcal{P}_{sf}(G[A \setminus I_{t}])$. We refer to $\ell(P_t)$ as the \emph{evader's instantaneous loss}. The evader also reveals the arcs in $P_t$ and their costs to the interdictor.

\item The interdictor updates the information available to him, i.e., $A_t = A_{t-1} \cup P_t$.
\end{enumerate}

We assume that the evader attempts to minimize her cumulative loss over $T$ rounds, while the interdictor is restricted to act greedily in each decision epoch. In addition, we make the following assumptions:
\vspace{2mm}

\textbf{A1.} In each round the interdictor acts first. Furthermore, the interdictor is \emph{greedy} in the sense that he always blocks a set of $k$-most vital arcs in the observed network, i.e.,
\begin{equation}\label{eq:greedy-interdictor}I_{t} \in \argmax \{z(G[A_{t-1} \setminus I]): \ I \subseteq A_{t-1}, |I|\leq k\}.\end{equation}

\looseness-1\textbf{A2.} Graph $G$ is not trivially $k$-separable that is, any subset of $k$ arcs in $G$ is not an $s$-$f$ cut.

\looseness-1\textbf{A3.} If there is more than one possible choice for $I_t$, then the interdictor blocks arcs following a well-defined deterministic rule, which is \emph{consistent} in the sense that if $I_t$ is chosen from a collection of blocking solutions $ {\cal I}$, then it is also chosen from any collection of solutions $\widetilde{\cal I}\subseteq {\cal I}$ containing $I_t$.


\textbf{A4.} The evader has full information about the graph's structure, costs  and the  interdictor's budget, $k$.  The evader observes the interdictor's actions before choosing a path and cannot use interdicted arcs.

\textbf{A5.} The interdictor is initially given information only about subnetwork $G[A_0]$. Each round he observes path $P_t$ and cost $c_a$ of each arc $a \in P_t$ used by the evader. 
\vspace{2mm}

The first part of Assumption \textbf{A1} is technical. The second part of \textbf{A1} can be motivated by the interdictor's incomplete initial information about the network's structure and costs, which is reasonable from the application perspective. More importantly, as we discuss explicitly in Section~\ref{sec:intro}, the studies in \cite{1,2} establish a number of attractive and practically relevant features of greedy interdiction policies both in the shortest path interdiction and general max-min linear mixed-integer programming settings.

Assumption \textbf{A2} is technical as it ensures that the evader's problem is feasible at each round. \looseness-1 Assumption \textbf{A3} implies that the interdictor's policies are deterministic. The \emph{consistency} assumption mimics an analogous assumption in \cite{1} for the evader's policies. For example, one can think that in each decision epoch the interdictor ranks all feasible blocking solutions in the observed network based on some criteria, e.g., their costs to the evader, resolving ties according to any deterministic criteria. Then the interdictor selects the highest-rank blocking solution from such a list.

\looseness-1 Assumption \textbf{A4} implies that the evader has some degree of monitoring of the interdictor's actions. The second part of this assumption is a standard conjecture in the interdiction, bilevel, and general hierarchical optimization literature; see, e.g., \cite{smith2008algorithms,smith2013modern,wood2011bilevel,sefair2016dynamic}.


The first part of Assumption \textbf{A5} formalizes the notion of some initial knowledge of the interdictor about the underlying network. The second part of Assumption \textbf{A5} represents the case of the perfect (or transparent) feedback (from the evader to the interdictor) similar to~\cite{1,2} that are the basis for this study (recall our discussion in Section~\ref{sec:intro}). We exploit this assumption in derivations of our theoretical results. However, we relax this assumption in our computational study in Section~\ref{sec:experiment}. Finally, recall that the interdictor has full information about the costs of arcs in $A_0$. Thus, \textbf{A5} implies that whenever existence of the arc is known to the interdictor at any decision epoch, then he is also aware of its cost.


In view of the discussion above, the evader's problem can be formulated as the following  repeated  hierarchical combinatorial optimization problem:
\allowdisplaybreaks\begin{subequations}\label{eq1}
\begin{align}
 \min_{P_t}\sum\limits_{t = 1}^{T}& \ell(P_t):=\sum\limits_{t = 1}^{T}\sum_{a\in P_t}c_a \label{eq1a}\\
\mbox{\upshape s.t. } & P_t \in \mathcal{P}_{sf}(G[A \setminus I_{t}]) \quad \forall t\in \{1,\ldots,T\}, \label{eq1b}\\
& I_{t} \in \argmax \{z(G[A_{t-1} \setminus I]): \ I \subseteq A_{t-1}, |I|\leq k\} \quad \forall t\in \{1,\ldots,T\},\label{eq1c}\\
& A_t = A_{t-1}\cup P_t \quad \forall t\in \{1,\ldots,T\}, \label{eq1d}
\end{align}
\end{subequations}
\looseness-1\noindent where the evader's objective function in (\ref{eq1a}) represents the sum of the evader's instantaneous losses of $T$ decision epochs. Condition (\ref{eq1b}) ensures that $P_t$ does not include arcs, which are blocked by the interdictor at round $t$. Constraint (\ref{eq1c}) requires $I_{t}$ to be a set of $k$-most vital arcs in $G[A_{t-1}]$ (recall assumption \textbf{A1}), i.e., an optimal solution of the interdiction problem in each decision epoch, and thus, leads to the hierarchical decision-making structure of the overall sequential problem. Then (\ref{eq1d}) states that a set of arcs known to the interdictor at round $t$ is updated according to assumption \textbf{A5}.

From the game theoretic perspective, the evader's problem (\ref{eq1}) can be viewed as a finitely repeated Stackelberg game with incomplete information; see, e.g., \cite{Sinha2018} and the references therein. Specifically, in each decision epoch $t \in \{1,\ldots, T - 1\}$ the evasion-blocking pair $(P_t, I_{t + 1})$ forms a Stackelberg equilibria in the currently observed network $G[A_t]$.  Next, we provide formal definitions of evasion and interdiction policies along with some examples.
\begin{definition} \label{def: evasion policy}
\upshape
An \textit{evasion policy} is a deterministic sequence of set functions $(\pi_t)|_{t\in \{1,\ldots,T\}}$ such that for each $t\geq 1$, $P^\pi_t = \pi_t(F^\pi_t)$, $P^\pi_t \in \mathcal{P}_{sf}(G[A \setminus I_{t}])$, where $F^\pi_t$ summarizes the initial information as well as the history of the interdiction and evasion decisions up to round $t$:
$$F^\pi_t = (\mathcal{C}_0, I_1, P^\pi_1, \ldots, I_{t-1}, P^\pi_{t-1}, I_{t}).$$
\vspace{-20mm}\flushright$\square$
\end{definition}
 Note that in Definition \ref{def: evasion policy} an evasion policy in round $t$ accounts for the sequence of previous blocking and evasion decisions. This assumption is necessary since the evader attempts to minimize her cumulative loss over $T$ rounds and thus, forms an overall sequence of evasion decisions $P^\pi_t$, $t \in \{1, \ldots, T\}$.

\begin{definition}
\upshape
The evader and her policy are referred to as \textit{greedy} if she chooses the shortest path in the interdicted network in each round. 
The evader and her policy are called \textit{strategic} if she solves (heuristically or exactly) sequential hierarchical optimization problem (\ref{eq1}).\vspace{-9mm}\flushright $\square$
\end{definition}

\looseness-1{Simply speaking, the greedy evader corresponds to a decision-maker who is myopic and optimizes only her instantaneous loss in each decision epoch. Alternatively, the upper-level decision-maker in (\ref{eq1}) can be viewed as a team of evaders, one for each decision epoch. Then the greedy policy corresponds to a scenario where the team is decentralized and each evader optimizes her own objective function, while the objective function in (\ref{eq1a}) represents the total team's loss.}

We denote by $P_t^{SP}$ and $P_t^{SE}$, $t\geq 1$, the evasion decisions by the greedy and strategic evaders, respectively. Given $G$ and $\mathcal{C}_0$, define the \emph{cumulative loss} of the evader under evasion policy $\pi$ over $T$ rounds as follows:
$$L_T^\pi(G,\mathcal{C}_0):=\sum_{t=1}^T \ell(P^\pi_t)$$
Let $\Pi(G, \mathcal{C}_0)$ be a class of all feasible evasion policies for $G$ and $\mathcal{C}_0$. Thus, for some fixed values of $k$ and $T$ we say that evasion policy $\pi \in \Pi(G,\mathcal{C}_0)$ \textit{strongly dominates} policy $\pi' \in \Pi(G,\mathcal{C}_0)$ if $L^{\pi}_T(G,\mathcal{C}_0) < L^{\pi'}_T(G,\mathcal{C}_0)$. 


\begin{definition} \label{def: greedy semi-oracle}
\upshape
\looseness-1 The interdictor is referred to as a \textit{greedy semi-oracle}, if for any $t \in \{1, \ldots, T\}$ the following conditions hold:
\begin{enumerate}
\item[(\textit{i})] The interdictor has complete knowledge of the evader's policy $\pi$.

\item[(\textit{ii})] If the $k$-most vital arcs problem in $G[A_{t-1}]$, i.e., problem (\ref{eq1c}), has multiple optimal solutions, then the interdictor selects the one that is least favorable for the evader under policy $\pi$ in the current decision epoch $t$.
\item[(\textit{iii})] Furthermore, among these solutions the one with maximum cardinality is preferred. \vspace{-11.5mm}\flushright $\square$

\end{enumerate}
\end{definition}

In other words, the greedy semi-oracle chooses a set of the $k$-most vital arcs $I_t \subseteq A_{t-1}$ in $G[A_{t-1}]$ so as to maximize the evader's instantaneous loss $\ell(P^\pi_t)$. We also provide a technical requirement that the interdictor maximizes $|I_t|$ as his auxiliary objective;  see, e.g., Example \ref{Example 1}.


Simply speaking we assume that a greedy semi-oracle has full information about the network's structure and costs, but is restricted to act greedily by blocking a set of the $k$-most vital arcs in the currently observed network $G[A_{t - 1}]$. However, given his knowledge of $G$ he can anticipate the evader's actions in the current round and thus, select solutions that are preferable for him and least favorable to the evader.

For example, assume that the policy $\pi$ in Definition \ref{def: greedy semi-oracle} is a greedy evasion policy and $T = 1$. Then our definition of the greedy semi-oracle can be viewed as a pessimistic version of the evader's problem (\ref{eq1}). More precisely, if the lower-level decision-maker (the interdictor) has multiple optimal solutions, then he prefers the one that is least favorable to the upper-level decision-maker (the evader); see~\cite{Colson2007} for further details on bilevel optimization.

Next, we provide two illustrative examples comparing the greedy evader against a strategic one. These examples provide us with further motivation of exploring the structural properties of strategic evasion policies discussed in Section~\ref{sec:theory}. In the examples we assume that the interdictor is a greedy semi-oracle.

\begin{example} \label{Example 1}
\upshape
Graph $G$ used in this example, is provided in Figure \ref{fig: Example 1}. Let $s = 1$ and $f = 4$ be the evader's source and destination nodes, respectively, and $M$ be a real number such that~$M > 5$. We also set $T = 2$, $k = 2$ and $A_0 = \emptyset$.

\looseness-1First, assume that the evader is greedy. Since $A_0 = \emptyset$ we have $I_1 = \emptyset$ and in the first round the greedy evader follows the overall shortest path given by $P^{SP}_1 =\{ 1 \rightarrow 2 \rightarrow 3 \rightarrow 4\}$. {Note that $A_1 =P^{SP}_1$ and thus, any subset of arcs of $P^{SP}_1$ is also an optimal solution of the $k$-most vital arcs problem in $G[A_1]$.}

 Next, recall that the interdictor is a greedy semi-oracle; see Definition \ref{def: greedy semi-oracle}. Hence, he knows that the evasion policy $\pi$ is greedy and attempts to maximize the evader's instantaneous loss $\ell(P^{\pi}_2) = z(G[A \setminus I_2])$ in the second decision epoch.  We conclude that the interdictor blocks arcs $(1, 2)$ and $(3,4)$, which implies that $I_2 = \{(1, 2), (3, 4)\}$ and $P^{SP}_2 = \{1 \rightarrow 5 \rightarrow 4\}$. Observe that $|I_2| = 2 = k$ and, therefore, the blocking solution with maximal possible cardinality is selected.  As a result, the cumulative loss of the greedy evader is given by $L^{SP}_2 = 3 + M$.

Then consider a strategic evader, who sequentially traverses through arc-disjoint paths $P^{SE}_1 =\{1 \rightarrow 2 \rightarrow 4\}$ and $P^{SE}_2 = \{1 \rightarrow 3 \rightarrow 4\}$. In this case $I_2 = \{(1, 2), (2,4)\}$ and the cumulative loss of the strategic evader is given by $L^{SE}_2 = 4 + 4 = 8 < L^{SP}_2$. \vspace{-10mm}\flushright$\square$\end{example}

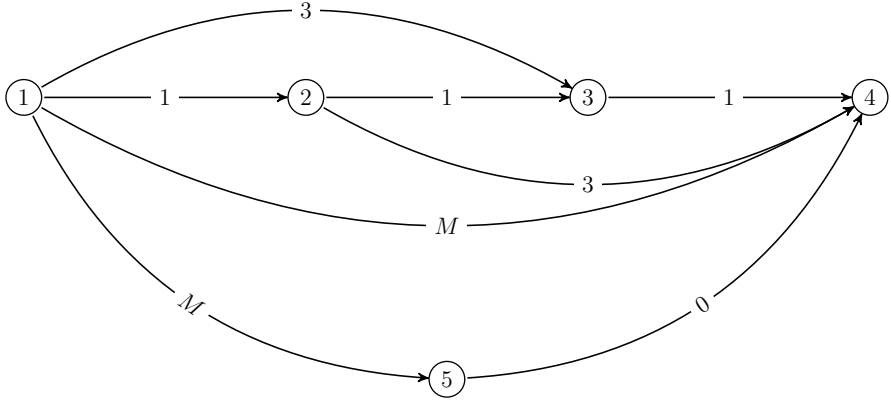
\begin{figure}
\centering
\begin{tikzpicture}[scale=0.75,transform shape]
\Vertex[x=0,y=0]{1}
\Vertex[x=5,y=0]{2}
\Vertex[x=10,y=0]{3}
\Vertex[x=15,y=0]{4}
\Vertex[x=7.5,y=-5]{5}
\tikzstyle{LabelStyle}=[fill=white,sloped]
\tikzstyle{EdgeStyle}=[post]
\Edge[label=$1$](1)(2)
\Edge[label=$1$](2)(3)
\Edge[label=$1$](3)(4)
\tikzstyle{EdgeStyle}=[post, bend left]
\Edge[label=$3$](1)(3)
\tikzstyle{EdgeStyle}=[post, bend right]
\Edge[label=$M$](1)(4)
\Edge[label=$3$](2)(4)
\Edge[label=$M$](1)(5)
\Edge[label=$0$](5)(4)
\end{tikzpicture}
\caption{\footnotesize {The network used in Examples \ref{Example 1} and \ref{Example 2}. For the costs of arcs $(1,4)$ and $(1,5)$ we assume that $M > 5$.}}
\label{fig: Example 1}
\end{figure}%

Example~\ref{Example 1} illustrates that if the evader is aware that the interdictor is greedy, then she can exploit this fact to decrease her cumulative loss. Furthermore, observe that the paths used by the strategic evader have some arcs in common with the shortest path. In Section~\ref{sec:theory} we formally establish that this observation is, in fact, a necessary condition for the strategic evasion policy whenever it outperforms the greedy one under the assumption that $T = 2$, $k \geq 1$ and $A_0 = \emptyset$.



In Example~\ref{Example 2} provided below we demonstrate that the greedy evasion policy can be dominated by a strategic policy for arbitrarily large $T \geq 2$, while $A_0$ does not necessarily need to be empty.


\begin{example} \label{Example 2}
\upshape
\looseness-1 As in the previous example consider the graph depicted in Figure~\ref{fig: Example 1}. In contrast to Example \ref{Example 1}, we change $A_0$ in a particular way and assume that $T$ is arbitrarily large, i.e., $T\geq 2$. More precisely, let $A_0$ be non-empty and given by $A_0 = \{(1, 3), (2, 4),(1,4), (1,5), (5,4)\}$.

First, we have $I_1 = \{(1,4), (1,5)\}$ regardless of the evasion policy as the interdictor acts first. Then the greedy evader traverses along path $P^{SP}_1 =\{ 1 \rightarrow 2 \rightarrow 3 \rightarrow 4\}$. Next, at $t=2$ the interdictor blocks $I_2 = \{(1, 2), (3, 4)\}$, which is a set of the $k$-most vital arcs in $G[A_1]$ and the evader follows $P^{SP}_2 =\{ 1 \rightarrow 5 \rightarrow 4\}$. Furthermore, observe that $P^{SP}_t = P^{SP}_2$, $I_t = I_2$ for all~$t \geq~3$. Hence, the cumulative loss of the greedy evader is given by $L^{SP}_T = 3 + M(T - 1)$ for any $T\geq 2$.

For a strategic evader, assume that she traverses through $P^{SE}_1 =\{ 1 \rightarrow 3 \rightarrow 4\}$ at $t=1$. Then $I_2 = \{(1, 4), (3, 4)\}$ and $P^{SE}_2 =\{ 1 \rightarrow 2 \rightarrow 4\}$. Consequently, $I_t = \{(1, 2), (3, 4)\}$ and $P^{SE}_t = P^{SP}_2$ for all $t \geq 3$. It implies that the cumulative loss of the strategic evader is $L^{SE}_T = 8 + M(T - 2)$. Therefore, for arbitrary $T\geq 2$ we have:
$$L^{SP}_T - L^{SE}_T = M - 5 > 0,$$
 and thus, the greedy evasion policy is suboptimal for arbitrarily large values of parameter $T$. $\square$
\end{example}

\section{Computational Complexity}\label{sec:NP}

Observe that the evader's problem in the case of $T = 1$ can be solved efficiently whenever the $k$-most vital arcs problem in $G[A_0]$ admits a polynomial time algorithm. Alternatively, the evader may have access to an \textit{interdiction oracle} that provides optimal blocking decisions in the network currently observed by the interdictor. If such interdiction oracle exists, then for $A_0 \neq \emptyset$ and $T = 1$ the evader's problem coincides with the shortest path problem in the interdicted network $G[A \setminus I_1]$, which is known to be polynomially solvable \cite{Ahuja}.

However, in general the $k$-most vital arcs problem is known to be $NP$-hard \cite{5} and thus, checking feasibility of an evasion solution is $NP$-hard. A more challenging problem is to determine whether the evader's problem is computationally difficult, when the interdiction problem is ``easy" to solve. In this section we show that the evader's problem is $NP$-hard in the case of $T = 2$ even for network instances where the interdiction problem can be solved in polynomial time. Henceforth, we assume that the interdictor is a greedy semi-oracle.

\looseness-1

%


In our complexity reduction below we assume that arc set $A$ consists of two disjoint subsets, namely, the arcs that are either \textit{removable} or \textit{unremovable} by the interdictor, respectively. The notion of unremovable arcs is technical and made without loss of generality as we can make all arcs removable by a polynomial time modification of the original graph. Specifically, we can simply replace each unremovable arc $(i,j)$ by $k + 1$ parallel arcs of equal costs. This construction guarantees that after the removal of at most $k$ arcs, at least one of these arcs remains intact and thus, nodes $i$ and $j$ remain connected by a directed arc. Alternatively (e.g., if parallel arcs are not allowed), we can replace each unremovable arc by $k+1$ paths of length two by adding $k+1$ and $2(k+1)$ of new nodes and arcs, respectively, which results in the same outcome.

Next, we define the classical \textbf{3-SAT} problem, which is known to be $NP$-complete~\cite{GJ}:\\

\noindent\textbf{Problem 3-SAT.}\\
\textit{Instance}: collection $F = \{F_1,\ldots,F_m\}$ of clauses on a finite set of variables $U$, $|U| = n$, such that $|F_i| = 3$ for $i \in \{1,\ldots,m\}$.\\
\textit{Question}: is there a truth assignment $\tau$ for $U$ that satisfies all the clauses in $F$? $\square$\\

Boolean formula $F = F_1\wedge F_2\wedge \ldots \wedge F_m$ is satisfied under assignment $\tau$ if and only if each clause is true. Clause $F_j$ is true if and only if it contains either literal $x_{i_j}$ such that $\tau (x_{i_j}) = true$ or literal $\overline{x}_{i_j}$ such that $\tau (\overline{x}_{i_j}) = false$. 

We also define the decision version of the evader's problem (\textbf{EP}) for $T = 2$:\\

\noindent \textbf{Problem 2-EP.}\\
\textit{Instance}: network $G$ together with source and destination nodes, subset of arcs $A_0$ known to the interdictor, the interdiction budget $k\in\mathbb{Z}_{>0}$ and threshold $h\in\mathbb{R}_{>0}$.\\
\textit{Question}: is there two paths $P_1, P_2 \in \mathcal{P}_{sf}(G)$ of total cost at most $h$ that can be traversed sequentially by the evader given that the interdictor is a greedy semi-oracle? $\square$\\

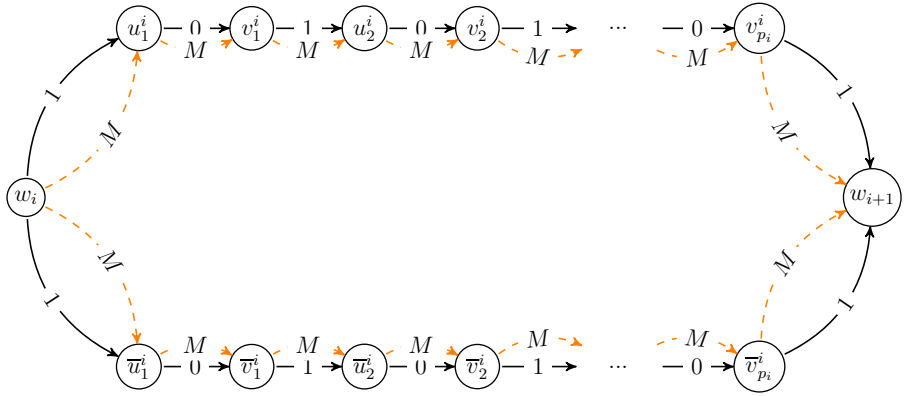
\begin{figure}
\centering
\begin{tikzpicture}[scale=0.75,transform shape]
\Vertex[x=0,y=0,L = $w_i $]{1}
\Vertex[x=15,y=0,L = $w_{i+1}$]{13}
\Vertex[x=2,y=3,L = $u^i_1$]{3}
\Vertex[x=2,y=-3,L = $\overline{u}^i_1$]{4}
\Vertex[x=4,y=3,L = $v^i_1$]{5}
\Vertex[x=4,y=-3,L = $\overline{v}^i_1$]{6}
\Vertex[x=6,y=3,L = $u^i_2$]{7}
\Vertex[x=6,y=-3,L = $\overline{u}^i_2$]{8}
\Vertex[x=8,y=3,L = $v^i_2$]{9}
\Vertex[x=8,y=-3,L = $\overline{v}^i_2$]{10}
\Vertex[x=13,y=3,L = $v^i_{p_i}$]{11}
\Vertex[x=13,y=-3,L = $\overline{v}^i_{p_i}$]{12}
\tikzstyle{VertexStyle}=[circle,fill=none]
\Vertex[x=10.5,y=3,L = $\quad...\quad$]{A}
\Vertex[x=10.5,y=-3, L = $\quad...\quad$]{B}
\tikzstyle{LabelStyle}=[fill=white,sloped]
\tikzstyle{EdgeStyle}=[post]
\Edge[label=$0$](3)(5)
\Edge[label=$1$](5)(7)
\Edge[label=$0$](7)(9)
\Edge[label=$0$](4)(6)
\Edge[label=$1$](6)(8)
\Edge[label=$0$](8)(10)
\Edge[label=$1$](10)(B)
\Edge[label=$1$](9)(A)
\Edge[label=$0$](A)(11)
\Edge[label=$0$](B)(12)
\tikzstyle{EdgeStyle}=[post, bend left]
\Edge[label=$1$](1)(3)
\Edge[label=$1$](11)(13)
\tikzstyle{EdgeStyle}=[post, bend left, dashed, orange]
\Edge[label=$M$](4)(6)
\Edge[label=$M$](6)(8)
\Edge[label=$M$](8)(10)
\Edge[label=$M$](10)(B)
\Edge[label=$M$](1)(4)
\Edge[label=$M$](12)(13)
\Edge[label=$M$](B)(12)
\tikzstyle{EdgeStyle}=[post, bend right]
\Edge[label=$1$](1)(4)
\Edge[label=$1$](12)(13)
\tikzstyle{EdgeStyle}=[post, bend right, dashed, orange]
\Edge[label=$M$](3)(5)
\Edge[label=$M$](5)(7)
\Edge[label=$M$](7)(9)
\Edge[label=$M$](9)(A)
\Edge[label=$M$](1)(3)
\Edge[label=$M$](11)(13)
\Edge[label=$M$](A)(11)
\end{tikzpicture}
\caption{\footnotesize Construction of the $i$-th lobe corresponding to variable $x_i$, where $p_i$ is the number of occurrences of variable $x_i$ in the clauses. Yellow dashed arcs are unremovable, and each of them has cost $M$ that is a sufficiently large constant parameter. The cost of removable arcs is either 0 or 1 as depicted for each arc.}
\label{pic:clause}
\end{figure}

\looseness-1The proof of our main result below is based on the reduction from the \textbf{3-SAT} problem, where for any instance of \textbf{3-SAT} we construct a particular instance of the \textbf{2-EP} problem. Following the discussion above we construct the instance of \textbf{2-EP} such that feasibility of any evasion solution can be checked in polynomial time with respect to the number of arcs, i.e., the $k$-most vital arcs problem in both $G[A_0]$ and $G[A_1]$ is polynomially solvable. Our construction is inspired and similar to the one used in \cite{18}, where it is shown that the problem of finding two minimum-cost arc-disjoint paths with non-uniform costs (e.g., changing over time or type of flow) is strongly $NP$-complete. However, our problem setting requires a somewhat different arc cost structure and the use of unremovable arcs defined at the beginning of this section.

Specifically, given boolean formula $F = F_1\wedge F_2\wedge\ldots\wedge F_m$, let $p_i$ be the number of occurrences of variable $x_i$ in $F$. For each variable $x_i$ we construct a \emph{lobe} as illustrated in Figure \ref{pic:clause}. 

The lobes are connected to one another in series with $w_1 = s$ and $w_{n+1} = f$. Recall that $s$ and $f$ are source and destination nodes, respectively. For each clause $F_j$, we add two nodes $y_j$, $z_j$, $j \in \{1,\ldots,m\}$ together with arcs $(s,y_1)$, $(z_j, y_{j+1})$, $j \in \{1,\ldots, m-1\}$ and $(z_m, f)$ of cost $0$. Finally, to connect clauses to variables we add the following arcs with zero costs: $(y_j, u^i_q)$ and $(v^i_q, z_j)$, if the $q$-th occurrence of variable $x_i$ is the literal $x_i$, which is a literal in clause $F_j$; $(y_j, \overline{u}^i_q)$ and $(\overline{v}^i_q, z_j)$, if the $q$-th occurrence of variable $x_i$ is the literal $\overline{x}_i$, which is a literal in clause $F_j$. We refer to Figure~\ref{pic:np-example} that illustrates the constructed graph for $F = (x_1\vee x_2 \vee x_3)\wedge(\overline{x}_1\vee \overline{x}_2 \vee x_3)\wedge(\overline{x}_1\vee x_2 \vee \overline{x}_3)$.

\begin{figure}
\centering
\begin{tikzpicture}[scale=0.75,transform shape]
\Vertex[x=0,y=0,L = $s$]{1}
\Vertex[x=7,y=0,L = $w_2$]{2}
\Vertex[x=14,y=0,L = $w_3$]{3}
\Vertex[x=21,y=0,L = $f$]{4}
\Vertex[x=2,y=-5,L = $y_1$]{5}
\Vertex[x=5,y=-5,L = $z_1$]{6}
\Vertex[x=9,y=-5,L = $y_2$]{7}
\Vertex[x=12,y=-5,L = $z_2$]{8}
\Vertex[x=16,y=-5,L = $y_3$]{9}
\Vertex[x=19,y=-5,L = $z_3$]{0}
\tikzstyle{VertexStyle} = [circle, inner sep=-2pt, minimum size=0pt, style = white, fill = black,sloped]
\Vertex[x=1,y=2, L=$ $]{10}
\Vertex[x=1,y=-2, L=$ $]{11}
\Vertex[x=2,y=2, L=$ $]{12}
\Vertex[x=2,y=-2,L=$ $]{13}
\Vertex[x=3,y=2,L=$ $]{14}
\Vertex[x=3,y=-2,L=$ $]{15}
\Vertex[x=4,y=2,L=$ $]{16}
\Vertex[x=4,y=-2,L=$ $]{17}
\Vertex[x=5,y=2,L=$ $]{18}
\Vertex[x=5,y=-2,L=$ $]{19}
\Vertex[x=6,y=2,L=$ $]{20}
\Vertex[x=6,y=-2,L=$ $]{21}

\Vertex[x=8,y=2,L=$ $]{22}
\Vertex[x=8,y=-2,L=$ $]{23}
\Vertex[x=9,y=2,L=$ $]{24}
\Vertex[x=9,y=-2,L=$ $]{25}
\Vertex[x=10,y=2,L=$ $]{26}
\Vertex[x=10,y=-2,L=$ $]{27}
\Vertex[x=11,y=2,L=$ $]{28}
\Vertex[x=11,y=-2,L=$ $]{29}
\Vertex[x=12,y=2,L=$ $]{30}
\Vertex[x=12,y=-2,L=$ $]{31}
\Vertex[x=13,y=2,L=$ $]{32}
\Vertex[x=13,y=-2,L=$ $]{33}

\Vertex[x=15,y=2,L=$ $]{34}
\Vertex[x=15,y=-2,L=$ $]{35}
\Vertex[x=16,y=2,L=$ $]{36}
\Vertex[x=16,y=-2,L=$ $]{37}
\Vertex[x=17,y=2,L=$ $]{38}
\Vertex[x=17,y=-2,L=$ $]{39}
\Vertex[x=18,y=2,L=$ $]{40}
\Vertex[x=18,y=-2,L=$ $]{41}
\Vertex[x=19,y=2,L=$ $]{42}
\Vertex[x=19,y=-2,L=$ $]{43}
\Vertex[x=20,y=2,L=$ $]{44}
\Vertex[x=20,y=-2,L=$ $]{45}

\tikzstyle{LabelStyle}=[fill=white,sloped]
\tikzstyle{EdgeStyle}=[post]
\Edge[](12)(14)
\Edge[](14)(16)
\Edge[](16)(18)
\Edge[](18)(20)

\Edge[](22)(24)
\Edge[](24)(26)
\Edge[](26)(28)
\Edge[](28)(30)

\Edge[](34)(36)
\Edge[](36)(38)
\Edge[](40)(42)
\Edge[](42)(44)

\Edge[label = $0$](5)(22)
\Edge[label = $0$](24)(6)
\Edge[label = $0$](5)(34)
\Edge[label = $0$](36)(6)

\Edge[label = $0$](7)(15)
\Edge[label = $0$](17)(8)
\Edge[label = $0$](7)(27)
\Edge[label = $0$](29)(8)

\Edge[label = $0$](9)(19)
\Edge[label = $0$](21)(0)
\Edge[label = $0$](9)(43)
\Edge[label = $0$](45)(0)

\tikzstyle{EdgeStyle}=[post, green, thick]
\Edge[](11)(13)
\Edge[](13)(15)
\Edge[](15)(17)
\Edge[](17)(19)
\Edge[](19)(21)

\Edge[](23)(25)
\Edge[](25)(27)
\Edge[](27)(29)
\Edge[](29)(31)
\Edge[](31)(33)

\Edge[](35)(37)
\Edge[](37)(39)
\Edge[](39)(41)
\Edge[](41)(43)
\Edge[](43)(45)

\tikzstyle{EdgeStyle}=[post, bend right, green, thick]
\Edge[](1)(11)
\Edge[](21)(2)
\Edge[](2)(23)
\Edge[](33)(3)
\Edge[](3)(35)
\Edge[](45)(4)

\tikzstyle{EdgeStyle}=[post, bend left]
\Edge[](1)(10)
\Edge[](20)(2)
\Edge[](2)(22)
\Edge[](32)(3)
\Edge[](3)(34)
\Edge[](44)(4)

\tikzstyle{EdgeStyle}=[post, bend right, blue, thick]
\Edge[label = $0$](1)(5)
\Edge[label = $0$](0)(4)

\tikzstyle{EdgeStyle}=[post, blue, thick]
\Edge[label = $0$](6)(7)
\Edge[label = $0$](8)(9)
\Edge[label = $0$](5)(10)
\Edge[](10)(12)
\Edge[label = $0$](12)(6)
\Edge[label = $0$](7)(38)
\Edge[](38)(40)
\Edge[label = $0$](40)(8)
\Edge[label = $0$](9)(30)
\Edge[](30)(32)
\Edge[label = $0$](32)(0)
\end{tikzpicture}
\caption{\footnotesize The graph corresponding to a \textbf{3-SAT} instance with $F = (x_1\vee x_2 \vee x_3)\wedge(\overline{x}_1\vee \overline{x}_2 \vee x_3)\wedge(\overline{x}_1\vee x_2 \vee \overline{x}_3)$. Here $n = 3$ is the number of lobes (or variables), $m = 3$ is the number of clauses, unremovable arcs and the arc costs within the lobes are not depicted (we refer to Figure \ref{pic:clause} for the detailed depiction of the lobes). Nodes $w_1=s$ and $w_4=f$ are the source and destination nodes of the constructed graph, respectively. Suppose $\tau(x_i) = 1$ for all $i \in \{1,\ldots,3\}$. Green and blue arcs form paths $P_1$ and $P_2$, respectively, i.e., the solution of the evader's problem with $T = 2$, $A_0 = \{(s, y_1),(z_m, f), (z_j, y_{j+1}), j \in \{1,\ldots,m - 1\}\}\cup B_0$, where $B_0$ is the set of unremovable arcs, $h = 3m + n$, $k$ is such that $3m \leq k \leq 6m + n$ and $M \geq 2$.}
\label{pic:np-example}
\end{figure}
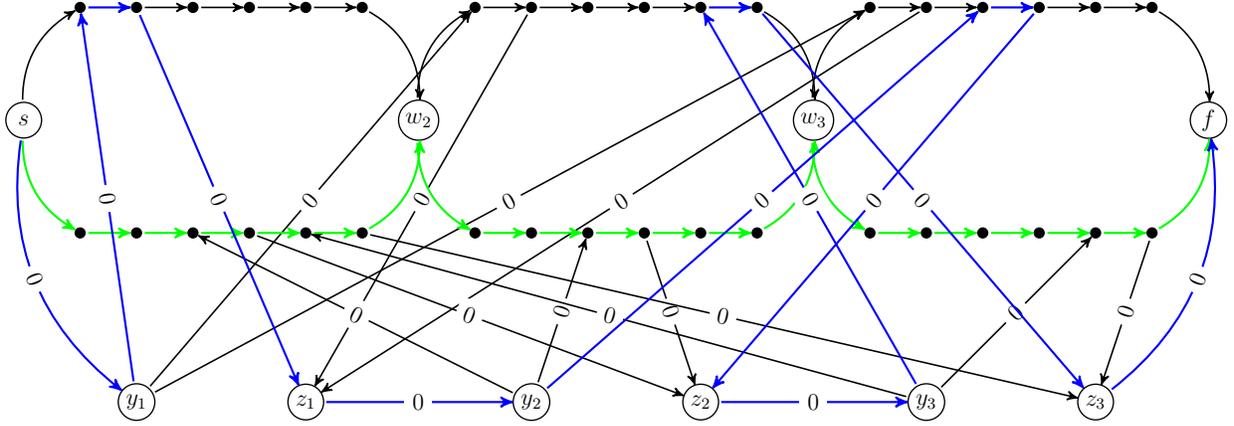
%

\begin{theorem} \label{th1}
Problem \textbf{2-EP} is strongly $NP$-complete for the class of network instances where the interdiction problem (in $G[A_0]$ and $G[A_1]$) is polynomially solvable.
\end{theorem}
\begin{proof}
$\Rightarrow$ Consider a ``yes''-instance of \textbf{3-SAT}. Thus, there exists assignment $\tau$ such that boolean formula $F = F_1\wedge F_2\wedge \ldots\wedge F_m$ is satisfied under $\tau$. Assume that we construct a graph for this instance as outlined in the discussion above. Next, let $A_0 = \{(s, y_1),(z_m, f), (z_j, y_{j+1}),\ j \in\{1,\ldots,m -1\}\}\cup B_0$, where $B_0$ is the set of unremovable arcs, $h = 3m + n$ and $k$ is such that $3m \leq k \leq 6m + n$, where $M \geq 2$.

Observe that $I_1 = A_0\setminus B_0$, since the interdictor is a greedy semi-oracle. Actually, $k$ is sufficiently large to block all arcs in $I_1$ and interdiction of $I_1$ maximizes the cost of the shortest path in $G[A_0\setminus I_1]$, i.e., $I_1 \in \argmax \{z(G[A_0\setminus I]): |I|\leq k, I\subseteq A_0\}$.

At $t = 1$ let the evader choose path $P_1$ that is constructed in the following way. It traverses through the lower part of the $i$-th lobe, if $\tau (x_i) = true$, and it traverses through the upper part, if $\tau (x_i) = false$. Observe that $\ell(P_1) = 3m + n$. The information collected by the interdictor is updated and thus, $A_1 = A_0 \cup P_1$. Note that set of arcs $A_1$ consists of path $P_1$ together with parallel unremovable arcs and set $A_0 \setminus B_0$ that contains $m + 1$ distinct arcs that do not form a path from $s$ to $f$. Furthermore, any evasion decision at $t = 1$ does not include arcs in $A_0$ and goes through all the lobes. We conclude that the interdiction problem in both $G[A_0]$ and $G[A_1]$ is polynomially solvable.

The presence of unremovable arcs of sufficiently large costs that are parallel to arcs in $P_1$ enforces the interdictor to remove arcs in the order of their costs. That is he blocks arcs of zero cost first and then blocks arcs of unit cost. Recall that $k \geq 3m$. Therefore, the blocking decision at $t = 2$ consists of at least $3m$ arcs of $P_1$ including all arcs with zero cost. We conclude that $I_2 \subseteq P_1$ and $I_2\cap I_1 = \emptyset$ due to the fact $k \leq |P_1| = 6m + n$.

Then there exist at least $m$ un-blocked subpaths $\{y_j \rightarrow u^i_k \rightarrow v^i_k \rightarrow z_j\}$ or $\{y_j \rightarrow \overline{u}^i_k \rightarrow \overline{v}^i_k \rightarrow z_j\}$ that correspond to the variable assignments in $\tau$. Together with arcs $(s, y_1)$ and $(z_m, f)$ they form path $P_2$ that is arc-disjoint with $P_1$ by their construction and can be traversed by the evader at $t = 2$. The cumulative loss of the evader equals $h = 3m + n$. Hence, the answer to the evader's problem is ``yes.''

$\Leftarrow$ Consider a ``yes''-instance of the evader's problem with $A_0 = \{(s, y_1),(z_m, f), (z_j, y_{j+1}), j \in \{1,\ldots,m - 1\}\}\cup B_0$, $h = 3m + n$ and $k$ is such that $3m \leq k \leq 6m + n$. Then there exist two paths that can be traversed sequentially by the evader and their total cost is at most $h$.

Observe that the choice of $I_1$ does not depend on the evader's actions. Path $P_1$ goes through either the upper or the lower part of each lobe. Therefore, path $P_1$ in the evader's decision has cost $h = 3m + n$, while $P_2$ has zero cost. Furthermore, the fact that $k \geq 3m$ implies that $P_2$ is arc-disjoint with $P_1$.

Next, we note that as the cost of $P_2$ is zero, then it needs to contain arcs $(s,y_1)$ and $(z_m,f)$. Furthermore, $P_2$ needs to traverse through arcs $(z_j,y_{j+1})$, $j \in \{1,\ldots,m - 1\}\}$ and each lobe through the zero cost arcs. We construct assignment $\tau$ in the following way: if $P_1$ traverses through the lower part of the $i$-th lobe, let $\tau (x_i) = true$ and, if it traverses through the upper part, then let $\tau (x_i) = false$. According to this assignment due to the existence of $P_2$ each clause $F_j$, $j \in \{1,\ldots,m - 1\}\}$, is satisfied, which implies the necessary result.
\end{proof}

\looseness-1{We conclude that the evader's problem is computationally hard even when $T = 2$ and the network instances admit a polynomial-time solution of the $k$-most vital arcs problem. The strong $NP$-completeness of \textbf{2-EP} also implies that there is no fully polynomial-time approximation scheme (or FPTAS) for the evader's problem unless $P = NP$~\cite{GJ}. For future work it can be of interest to explore whether constant factor approximation algorithms exist for \textbf{2-EP} and how the complexity of the problem changes if $k$ is bounded by a constant.}

{In view of the above in the remainder of the paper we explore analytical properties of the optimal evasion policies in a simple case, i.e., $T = 2$ and $A_0 = \emptyset$, and then exploit them to propose a heuristic algorithm for an arbitrary set of initial information and any time horizon that outperforms a greedy approach.}


\section{Basic Analysis of Evasion Policies}\label{sec:theory}

In this section we provide basic analysis of evasion policies for the case of two decision epochs when no initial information is available to the interdictor, i.e., $T = 2$ and $A_0 = \emptyset$. Specifically, under some rather mild assumption we show that an optimal solution of the evader's problem is either greedy or consists of two distinct paths that intersect with the shortest path in the non-interdicted network. The latter observation gives us a basis for developing a heuristic approach of finding strategic evasion decisions; see Section~\ref{sec:algorithm}. Then we demonstrate that for $k = 1$ the greedy evasion policy is optimal. Finally, it can be rather easily shown that for sufficiently large values of $k$ an optimal evasion solution consists of two arc-disjoint paths.

%
In the remainder of this section to {simplify our derivations} we assume that the costs of all possible paths from $s$ to $f$ are distinct. Thus, we can enumerate them in the strictly increasing order of their costs, i.e., \begin{equation}\label{eq:analysis-assumption}\ell(P^{(1)}) < \ell(P^{(2)}) < \ldots < \ell(P^{(\mu)}),\end{equation}
\noindent where $|\mathcal{P}_{sf}(G)| = \mu$. Denote by $\nu(P)$ the index of path $P \in \mathcal{P}_{sf}(G)$ in the above ordering.

\begin{lemma} \label{remark paths}
\upshape
Let $A_0 = \emptyset$, $k \geq 1$ and assume that the evader is greedy. Denote by $r$ the index of path traversed by the evader at $t = 2$, i.e., $r = \nu(P^{SP}_2)$. Then for any $z \in \{1,2,\ldots,r - 1\}$ path $P^{(z)}$ is blocked by the interdictor at $t = 2$.
\end{lemma}
\begin{proof}
The results follows from the definition of the greedy evader. Since $A_0 = \emptyset$ and $k \geq 1$, we have $r > 1$.
\end{proof}
 Next, we provide some basic necessary conditions for an evasion solution to be optimal when $T = 2$ and $A_0 = \emptyset$.
\begin{theorem} \label{proposition2}
Assume that the interdictor is a greedy semi-oracle with $A_0 = \emptyset$ and $k \geq 1$. Let $T = 2$ and $r = \nu(P^{SP}_2)$, where $r\geq 2$. If $P^{OPT}_1 = P^{(i)}$ and $P^{OPT}_2 = P^{(j)}$ is an optimal solution of the evader's problem for $T = 2$, then either $i = 1$ and $j = r$, or $P^{(i)}$ and $P^{(j)}$ satisfy the following conditions:
\begin{align}
P^{(i)} \cap P^{(1)} \neq \emptyset \quad \mbox{and} \quad P^{(j)} \cap P^{(1)} \neq \emptyset, \label{eq2} \\
\ell(P^{(1)}) + \ell(P^{(r)}) > \ell(P^{(i)}) + \ell(P^{(j)}) ,\label{eq3}\\
1 < i < r, \ 1 < j < r \quad \mbox{and} \quad i \neq j. \label{eq4}
\end{align}
\end{theorem}
\begin{proof}
First, suppose that the greedy evasion policy is optimal. By the definition of $r$ we have that the greedy evader traverses through $P^{(1)}$ and $P^{(r)}$ sequentially. Hence, $i = 1$ and $j = r$.

Conversely, assume that the optimal evasion policy is not greedy, which implies that condition (\ref{eq3}) holds. Thus, we need to show that conditions (\ref{eq2}) and (\ref{eq4}) are satisfied.

If condition (\ref{eq4}) does not hold, then we consider the following possible scenarios:
\begin{itemize}
\item[(\emph{i})]
We have that either $i\geq r$, or $j\geq r$. If $i \geq r$, then by assumption (\ref{eq:analysis-assumption}) the following inequalities hold:
$$\ell(P^{(1)}) \leq \ell(P^{(j)}),\quad \ell(P^{(r)}) \leq \ell(P^{(i)}),$$
\noindent which implies that (\ref{eq3}) is violated and results in a contradiction. The case $j \geq r$ is similar;
\item[(\emph{ii})]
We have that $i = 1$. Clearly in this case the greedy evasion solution is optimal and (\ref{eq3}) is not satisfied (recall that the greedy evader traverses through $P^{(1)}$ at $t = 1$);
\item[(\emph{iii})]
We have that $j = 1$. In other words, path $P^{(1)}$ is not blocked by the interdictor at $t = 2$. Since the interdictor is a greedy semi-oracle and the evader traverses through $P^{(i)}$ and $P^{(1)}$ sequentially we conclude that $P^{(i)}$ and $P^{(1)}$ are arc-disjoint, i.e., $P^{(i)} \cap P^{(1)} = \emptyset$. Thus, conversely, if the evader is greedy, then path $P^{(i)}$ cannot be blocked by the interdictor at $t = 2$ as $A_0 = \emptyset$. From Lemma \ref{remark paths} we observe that $i\geq r$ and thus, we can simply refer to the discussion above in (\emph{i}).
\end{itemize}
We conclude that condition (\ref{eq4}) is satisfied.

Finally, assume that (\ref{eq2}) is violated and thus, either $P^{(i)}$ or $P^{(j)}$ is arc-disjoint with $P^{(1)}$. Hence, if the evader is greedy and $A_0 = \emptyset$, then either $P^{(i)}$ or $P^{(j)}$, respectively, cannot be blocked at $t = 2$. Therefore, we have that either $i \geq r$ or $j \geq r$ (recall again Lemma \ref{remark paths}), which contradicts (\ref{eq4}) and hence, completes the proof.
\end{proof}


This result implies that there exist two mutually exclusive alternatives: either the greedy evasion policy is optimal or there exists another evasion solution (better than the greedy one), which consists of two distinct paths that both have some arcs in common with the shortest path from $s$ to $f$ in the non-interdicted graph, see (\ref{eq2}).


Clearly, one natural question arising next is whether these two distinct paths of the latter alternative are either arc-disjoint or intersect by themselves as well. In particular, from Example~\ref{Example 1} in Section \ref{sec:mathmod} one may hypothesize that an optimal evasion solution for $T = 2$ and $A_0 = \emptyset$ is either greedy, or consists of two arc-disjoint paths of minimal total costs. While it is often the case in many simple problem instances, nevertheless this hypothesis does not hold in general. Specifically, next we construct an instance of the evader's problem, where an optimal evasion solution consists of two distinct paths that intersect along some arcs that are also contained in the shortest path from $s$ to $f$. The latter fact illustrates that the necessary conditions given by (\ref{eq2})-(\ref{eq4}) in Theorem \ref{proposition2} hold.


\begin{example} \label{counterexample} \upshape
Consider the network depicted in Figure \ref{graph_counterexample}. Let $s = 1$ and $f = 11$. Assume also that $T = 2$, $k = 3$ and $A_0 = \emptyset$. Observe that in the first decision epoch the greedy evader follows the shortest path $P^{SP}_1 = \{1 \rightarrow 2 \rightarrow \ldots \rightarrow 11\}$ of cost zero, while the interdictor blocks set of arcs $I_1 = \{(1,2), (4,5), (10,11)\}$. Hence, in the second round the evader must traverse through path $P^{SP}_2 = \{1 \rightarrow 12 \rightarrow 11\}$ with her total loss of $L^{SP}_2 = 0 + 11 = 11$ over two rounds.

Next, it can be verified that the cost of any arc-disjoint evasion solution from $s$ to $f$ is at least $11$. Nevertheless, consider a strategic evader who first traverses through path $P^{SE}_1 = \{1 \rightarrow 3 \rightarrow 4 \rightarrow \ldots \rightarrow 11\}$. Then the interdictor's solution is $I_1 = \{(4, 5), (6, 7), (8, 9)\}$ to block all paths of cost $8$. Then the strategic evader follows path $P^{SE}_2 = \{1 \rightarrow 10 \rightarrow 11\}$ and her cumulative loss over two rounds is given by $L^{SE}_2 = 1 + 9 = 10 < L^{SP}_2$. Thus, in this example neither greedy nor arc-disjoint evasion solutions are optimal. \vspace{-10mm}\flushright$\square$
\end{example}
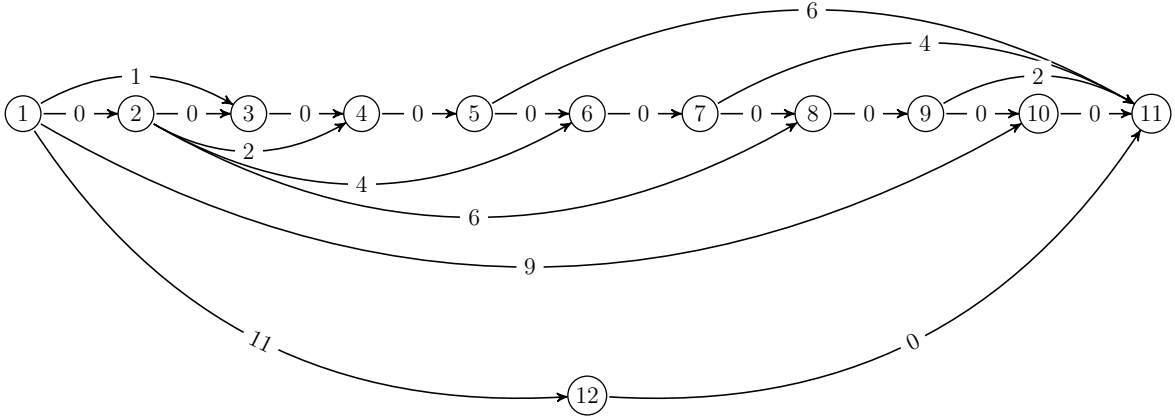
\begin{figure}
\centering
\begin{tikzpicture}[scale=0.75,transform shape]
\Vertex[x=0,y=0]{1}
\Vertex[x=2,y=0]{2}
\Vertex[x=4,y=0]{3}
\Vertex[x=6,y=0]{4}
\Vertex[x=8,y=0]{5}
\Vertex[x=10,y=0]{6}
\Vertex[x=12,y=0]{7}
\Vertex[x=14,y=0]{8}
\Vertex[x=16,y=0]{9}
\Vertex[x=18,y=0]{10}
\Vertex[x=20,y=0]{11}
\Vertex[x=10,y=-5]{12}
\tikzstyle{LabelStyle}=[fill=white,sloped]
\tikzstyle{EdgeStyle}=[post]
\Edge[label=$0$](1)(2)
\Edge[label=$0$](2)(3)
\Edge[label=$0$](3)(4)
\Edge[label=$0$](4)(5)
\Edge[label=$0$](5)(6)
\Edge[label=$0$](6)(7)
\Edge[label=$0$](7)(8)
\Edge[label=$0$](8)(9)
\Edge[label=$0$](9)(10)
\Edge[label=$0$](10)(11)
\tikzstyle{EdgeStyle}=[post, bend left]
\Edge[label=$1$](1)(3)
\Edge[label=$6$](5)(11)
\Edge[label=$4$](7)(11)
\Edge[label=$2$](9)(11)
\tikzstyle{EdgeStyle}=[post, bend right]
\Edge[label=$9$](1)(10)
\Edge[label=$2$](2)(4)
\Edge[label=$4$](2)(6)
\Edge[label=$6$](2)(8)
\Edge[label=$11$](1)(12)
\Edge[label=$0$](12)(11)
\end{tikzpicture}
\caption{\footnotesize The network used in Example \ref{counterexample}. The arc costs are depicted inside each arc.}
\label{graph_counterexample}
\end{figure}

\looseness-1However, both greedy and arc-disjoint evasion solutions can be shown to be optimal for specific values of parameter $k$. First, it can be shown that the greedy evasion policy is optimal when $k = 1$, i.e., the interdictor blocks one arc in each decision epoch.


\begin{proposition} \label{corrollary2}
Assume that the interdictor is a greedy semi-oracle with $A_0 = \emptyset$. Let $T=2$ and $k = 1$. Then the greedy evasion solution is optimal.

\begin{proof}
Recall from Lemma \ref{remark paths} that if the evader is greedy and $A_0 = \emptyset$, then for any $z \in \{1,2,\ldots,r - 1\}$ path $P^{(z)}$ is blocked at $t = 2$ by the interdictor. Since the interdictor's budget $k = 1$ we conclude that all paths $P^{(z)}$, $z \in \{1,2,\ldots,r - 1\}$, have an arc in common, i.e., there exist $e \in A$ such that $e \in P^{(z)}$ for all $z \in \{1,2,\ldots,r - 1\}$.

\looseness-1 Assume that the greedy evasion decision is not optimal. Then the optimal solution, namely, paths $P^{(i)}$ and $P^{(j)}$ must satisfy conditions (\ref{eq2}) - (\ref{eq4}) of Theorem \ref{proposition2}. In particular, we have that $1 < i,j < r$ and thus, these two paths have arc $e$ in common that is, $e \in P^{(i)} \cap P^{(j)}$. Following the discussion above arc $e$ also belongs to the shortest path $P^{(1)}$. Observe that with $A_1 = P^{(i)}$ the greedy semi-oracle interdictor must block $e$ and, thus, $P^{(j)}$ cannot be traversed by the strategic evader at $t = 2$. It contradicts with our earlier assumption that $P^{(i)}$ and $P^{(j)}$ form an optimal evasion solution.
\end{proof}
\end{proposition}


\begin{remark} \label{corrollary3}
\upshape
\looseness-1 Moreover, if $k$ is sufficiently large, then for $T=2$ it is rather easy to observe that if the interdictor is a greedy semi-oracle with $A_0 = \emptyset$, then an optimal evasion solution consists of two arc-disjoint paths. Specifically, since $A_0$ is empty, then in the first decision epoch there are no interdicted arcs. Consider an evasion path in the first decision epoch of any optimal evasion solution. If the value of $k$ is larger than or equal to the number of arcs in this path, then by the definition of a greedy semi-oracle all arcs in this path are interdicted in the second decision epoch. Thus, the evader's consequent evasion path cannot contain arcs from her previous evasion path.
%
\vspace{-9mm}\flushright$\square$\end{remark}

Next, recall that in Example \ref{Example 1} the loss of the greedy evader is equal to $3 + M$, while the optimal evasion solution has cost $8$. Hence, for any constant $\beta \in \mathbb{R}_{>0}$ and sufficiently large values of parameter $M$ we have:
$$\beta L^{SE}_2 < L^{SP}_2,$$
which implies that in general the greedy evasion solution cannot approximate an optimal evasion solution with any constant factor. The next example demonstrates that the latter property also holds for any arc-disjoint evasion solution.


\begin{example} \label{Example 4}
\upshape
Consider the network depicted in Figure \ref{fig: Example 4}. Assume that $T = 2$, $A_0 = \emptyset$ and $k = 2$. The greedy evader follows the shortest path of cost zero in the first decision epoch, i.e., $P^{SP}_1 = \{1 \rightarrow 2 \rightarrow \ldots \rightarrow |V|\}$. The interdictor blocks any $k = 2$ arcs of $P^{SP}_1$ and thus, the evader must use two arcs of cost $M$ in the second round. Observe that $L^{SP}_2 = 0 + 2M = 2M$.

On the other hand, any arc-disjoint solution contains all arcs of cost $M$ and thus, its cost is equal to $M \frac{|A|}{2}$. We conclude that the approximation factor of any arc-disjoint solution is $O(|A|)$ and therefore, not bounded by a constant.
\vspace{-10mm}\flushright$\square$
\end{example}

\begin{figure}
\centering
\begin{tikzpicture}[scale=0.75,transform shape]
\Vertex[x=0,y=0]{1}
\Vertex[x=4,y=0]{2}
\Vertex[x=12,y=0,L=\footnotesize $|V| - 1$]{n-1}
\Vertex[x=16,y=0,L=\footnotesize $|V|$]{n}
\tikzstyle{VertexStyle}=[circle,fill=none]
\Vertex[x=8,y=0,L = $\quad...\quad$]{A}
\tikzstyle{LabelStyle}=[fill=white,sloped]
\tikzstyle{EdgeStyle}=[post]
\Edge[label=$0$](1)(2)
\Edge[label=$0$](2)(A)
\Edge[label=$0$](A)(n-1)
\Edge[label=$0$](n-1)(n)
\tikzstyle{EdgeStyle}=[post, bend left]
\Edge[label=$M$](1)(2)
\Edge[label=$M$](2)(A)
\Edge[label=$M$](A)(n-1)
\Edge[label=$M$](n-1)(n)
\Edge[label=$M$](n-1)(n)
\tikzstyle{EdgeStyle}=[post, bend right]
\end{tikzpicture}
\caption{\footnotesize The network used in Example \ref{Example 4}. The arc costs are depicted inside each arc, where $M > 0$.}
\label{fig: Example 4}
\end{figure}
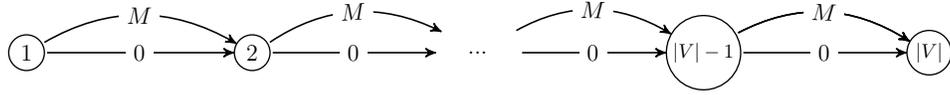

\looseness-1To summarize the discussion of this section, we conclude that in general both greedy and arc-disjoint evasion solutions can be suboptimal. Moreover, an optimal evasion solution may be somewhat non-trivial even for $T = 2$ and $A_0 = \emptyset$ as long as $k$ is at least two (recall Proposition~\ref{corrollary2} for $k=1$), but also not too large (Remark~\ref{corrollary3}). However, the necessary conditions given by Theorem~\ref{proposition2} provide us an intuition that we can exploit in the design of a heuristic algorithm for the strategic evader, which we discuss next.%



%

\section{Heuristic Algorithm for the Strategic Evader}\label{sec:algorithm}

\looseness-1 In this section we propose a heuristic algorithm for the strategic evader. The key idea of the algorithm is motivated by theoretical observations provided in the previous section, namely, Theorem \ref{proposition2}, and is based on the two-step ``look-ahead'' concept. {However, in contrast to Section~\ref{sec:theory} we do not assume that $A_0=\emptyset$.}

\begin{algorithm}
\DontPrintSemicolon
\textbf{Input:} Directed graph $G=(N, A, C)$, the source and destination nodes for the evader, the maximum number of arcs that can be blocked by the interdictor $k$, index of the current time epoch $\tau$, set of arcs $A_{\tau - 1}$ that are known to the interdictor, heuristic parameters $\alpha \in (0,1]$ and $q\in\mathbb{Z}_>0$\\
\textbf{Output:} Evasion decisions for the next one or two time epochs\\ 
$t \longleftarrow \tau, \quad I_{\tau} \longleftarrow \argmax \{z(G[A_\tau\setminus I]):\ |I| \leq k\}$\;
$\mathcal{P}^{SP}_\tau$ $\longleftarrow$ the shortest path in $G[A \setminus I_{\tau}]$\;
$A^{'}_\tau \longleftarrow A_{\tau - 1}\cup \mathcal{P}^{SP}_\tau$\;
$I^{'}_{\tau + 1} \longleftarrow \argmax \{z(G[A^{'}_\tau\setminus I]:\ |I| \leq k\}$\;
$\mathcal{P}^{SP}_{\tau + 1} \longleftarrow$ the shortest path in $G[A\setminus I^{'}_{\tau + 1}]$\;
$CurrentBestDecision \longleftarrow \{\mathcal{P}^{SP}_\tau\}$ \/\/ current (greedy) evasion decision\;
$CurrentBestSolution \longleftarrow \ell(\mathcal{P}^{SP}_{\tau}) + \ell(\mathcal{P}^{SP}_{\tau + 1})$\;
$\delta_\tau \longleftarrow \alpha\cdot \left(\ell(\mathcal{P}^{SP}_{\tau}) + \ell(\mathcal{P}^{SP}_{\tau + 1})\right)$ \/ a threshold that depends on the heuristic parameter~$\alpha$\;
\textbf{for} $(b^{(1)},\ldots,b^{(q)}) \in \mathcal{P}^{SP}_\tau$ \textbf{do}\\
\Begin
{
 $P^{SE}_\tau \longleftarrow$ the shortest path in $G[A\setminus (I_{\tau} \cup (\bigcup_{i = 1}^q b^{(i)})$]\;
 \textbf{if} $\ell(P^{SE}_\tau) < \delta_\tau$ and $P^{SE}_\tau \cap \mathcal{P}^{SP}_\tau \neq \emptyset $ \textbf{then}\\
 \Begin
 {
 $A^{''}_\tau \longleftarrow A_{\tau - 1}\cup P^{SE}_\tau$\;
 $I^{''}_{\tau + 1} \longleftarrow \argmax \{z(G[A^{''}_\tau\setminus I]:\ |I| \leq k\}$\;
 $P^{SE}_{\tau + 1} \longleftarrow$ the shortest path in $G[A\setminus I^{''}_{\tau + 1}]$\;
 \textbf{if} $\ell(P^{SE}_\tau) + \ell(P^{SE}_{\tau + 1}) < CurrentBestSolution$ \textbf{then}\\
 \Begin{
 $CurrentBestDecision \longleftarrow \{P^{SE}_\tau, P^{SE}_{\tau + 1}\}$\; $CurrentBestSolution \longleftarrow \ell(P^{SE}_\tau) + \ell(P^{SE}_{\tau + 1})$\;
 }
 }
}
\Return{$CurrentBestDecision$}
\caption{{Heuristic Algorithm for the Strategic Evader}}
\label{alg1}
\end{algorithm}

In the algorithm in each decision epoch the evader has two options: either she follows the greedy policy, or seeks for two alternative paths that can be traversed sequentially with the cumulative cost less than the loss obtained by the greedy approach. The pair of alternative paths is generated in a heuristic manner (see the details further in this section) by blocking a subset of arcs of some predefined cardinality of the shortest path in the network and then {verifying that conditions (\ref{eq2})-(\ref{eq3}) of Theorem \ref{proposition2} hold}. It is important to note that because of its two-step look-ahead scheme our algorithm can be applied in an iterative manner for any arbitrary time horizon~$T$.

We assume that the heuristic has access to an interdiction oracle that returns a set of $k$-most vital arcs in the network, currently observed by the interdictor. In general, the $k$-most vital arcs problem is known to be $NP$-hard \cite{12}. However, it can be rather effectively solved by decomposition algorithms; see, e.g., \cite{5}, that are typically faster than a naive branch-and-bound approach. Moreover, taking into account the structure of the problem, the decomposition algorithms can often be speeded up, if we store previous evasion decisions. Next, we discuss the key features of our algorithm in details and describe its links to Theorem \ref{proposition2}. We refer to Algorithm \ref{alg1} for the pseudocode.

\looseness-1 At the beginning of the $\tau$-th epoch the evader observes the information collected by the interdictor up to time $t = \tau$, i.e., set $A_{\tau - 1}$ and the corresponding blocking decision $I_{\tau}$. First, in lines $6 - 9$ of Algorithm \ref{alg1} we compute the evader's loss if she follows the greedy policy. (Note that we use $\mathcal{P}^{SP}_\tau$ and $\mathcal{P}^{SP}_{\tau + 1}$ to denote these paths as they do not necessarily coincide with ${P}^{SP}_\tau$ and ${P}^{SP}_{\tau+1}$ due to the iterative nature of the heuristic.)

\looseness-1 Then in lines $13-24$ we seek for a pair of alternative paths that can provide a better decision for the evader, e.g., if $T = 2$ and $A_0 = \emptyset$, then our construction ensures that these paths satisfy conditions (\ref{eq2})-(\ref{eq3}) of Theorem \ref{proposition2}. However, recall that Theorem \ref{proposition2} holds for $A_0 = \emptyset$ and does not take into consideration problems with $A_0 \neq \emptyset$. The latter is a more challenging case as the properties of alternative paths depend on the set of initial information; see, e.g., Example~\ref{Example 2}. Therefore, we resort to generating these paths in a heuristic manner as detailed below.

Specifically, to construct the first alternative path we first block $q\in\mathbb{Z}_{>0}$ arcs of $\mathcal{P}^{SP}_\tau$ together with set $I_{\tau}$ and identify the shortest path in the interdicted network, which is assumed to be the first alternative path; see line $13$ of Algorithm \ref{alg1}. As we block $q$ arcs of the shortest path, $\mathcal{P}^{SP}_\tau$, the newly constructed path, referred to as $\mathcal{P}^{SE}_\tau$, does not coincide with $\mathcal{P}^{SP}_\tau$.

\looseness-1 In line $14$ we verify, first, whether $\mathcal{P}^{SE}_\tau$ contains some arcs in common with $\mathcal{P}^{SP}_\tau$ and, furthermore, whether its loss is sufficiently small. {The former requirement is motivated by condition~(\ref{eq2}) of Theorem \ref{proposition2}}, while the latter requirement is controlled via the heuristic parameter $\alpha$ and is linked to the evader's loss under the greedy policy; see line $14$ of Algorithm \ref{alg1} and our additional comments below. If both of these requirements are satisfied, then in line $18$ we generate the second alternative path, $\mathcal{P}^{SE}_{\tau+1}$, in a greedy manner by assuming that $\mathcal{P}^{SE}_\tau$ is traversed at $t = \tau$. In line 19 we check whether the constructed solution is better than the best available solution, which coincides with a greedy one at the first iteration; see lines 8 and 9. Then we repeat the overall procedure in an iterative manner by blocking another subset of $q$ arcs of $\mathcal{P}^{SP}_\tau$.

\looseness-1 {As briefly outlined in the discussion above the developed heuristic has two tunable parameters: $\alpha$ and $q$. The first one, $\alpha \in (0,1]$, is used to compute $\delta_\tau$, which can be viewed as some measure of difference between the instantaneous losses of the greedy evasion paths in epochs $\tau$ and $\tau+1$ and a potential evasion path for the strategic evader in epoch $\tau$. Note that if $\alpha=0.5$, then $\delta_\tau$ is simply the average instantaneous loss of the greedy evader over two decision epochs.}

\looseness-1{In particular, for smaller values of parameter $\alpha$ the algorithm takes into consideration for the strategic evader only very promising evasion paths in epoch $\tau$, see the first condition in line 14. On the other hand, for larger values of $\alpha$ the algorithm considers a larger subset of candidate evasion paths that are also required to intersect with the shortest path in the interdicted network, see the second condition in line 14 and recall (\ref{eq2}). Thus, by increasing the value of $\alpha$ we increase the running time of the heuristic, but potentially improve the quality of the obtained myopic solution for the two-steps look-ahead approach. Also, note that improvements over two decision epochs do not necessarily imply improvements over longer decision-making horizons.}

\looseness-1 Furthermore, we use parameter $q$ to characterize the number of arcs of $\mathcal{P}^{SP}_\tau$
that are blocked in addition to $I_\tau$ in order to ensure that the first of the newly constructed alternative paths does not coincide with $\mathcal{P}^{SP}_\tau$. The choice of $q$ is related with the following trade-off. {Consider the case of $T = 2$ and $A_0 = \emptyset$. Note that Theorem \ref{proposition2}, see (\ref{eq2}), implies that the alternative evasion paths have some arcs in common with the shortest path (i.e., the initial evasion path of the greedy evader when $A_0 = \emptyset$). Intuitively, by increasing the value of $q$ we attempt to minimize the cardinality of these subsets of arcs in common (and thus, generate a path that is sufficiently different from the greedy evasion path in epoch $\tau$), see line 13. On the other hand, the loss of the first alternative path is required to be less than threshold $\delta_\tau$ (see line 14 of Algorithm~\ref{alg1}), which, in turn, requires a sufficiently small value of $q$.}

In summary, we observe that our heuristic approach finds optimal solutions for all the considered examples, namely, Examples \ref{Example 1}, \ref{Example 2} and \ref{Example 4}, under a suitable choice of the heuristic parameters, see additional discussion on the parameter settings in Section~\ref{subsec:prelim}. Furthermore, if $q$ is bounded by a constant, then Algorithm \ref{alg1} requires a polynomial in $|A|$ number of calls to the interdiction oracle.

\section{Computational Study}\label{sec:experiment}

In this section we compare the heuristic for the strategic evader (Algorithm \ref{alg1}) against the greedy policy. We show that the heuristic approach performs sufficiently well even for rather large values of $T$ (recall that Algorithm \ref{alg1} is myopic as it generates evasion decision only for at most two time epochs) and consistently outperforms the greedy policy on several classes of synthetic network instances. {In addition to the perfect feedback scenario, we also consider a noisy feedback scenario where for each arc initially unknown to the interdictor but traversed by the evader at decision epoch $t$, i.e., $a \in P_t \setminus A_{t - 1}$, the interdictor does not obtain the perfect information about the actual arc cost but rather observes a noisy realization of its nominal cost. In other words, we relax Assumption \textbf{A5} to reflect more realistic interdiction scenarios.}

The remainder of this section is organized as follows. First, in Section~\ref{subsec:prelim} we describe our test instances and the implementation issues. Then in Section~\ref{subsec:results} we discuss the obtained computational results.

\subsection{Preliminaries}\label{subsec:prelim}

\textbf{Graph structure and costs.} The test instances used in our experiments are represented by three classes of random graphs, specifically, layered, uniform~\cite{26} and Barabasi-Albert~(BA)~\cite{barabasi1999emergence} graphs. The test instances are constructed as follows:
\begin{itemize}
\item
\textit{Layered graphs.} Each of our random layered graphs consists of $h$ layers with $r_i$ nodes in the $i$-th layer, where
$r_i$ for $i\in\{2,\ldots,h-1\}$ is generated according to the discrete uniform distribution in the interval $[r_{min}, r_{max}]$. The first and last layers consist of a single node, i.e., $r_1=r_h=1$ , that are the source and destination nodes, respectively. An arc between a pair of nodes from the $i$-th and $j$-th layers is generated with probability $\frac{p}{j-i}$. Furthermore, the source node is connected by a directed arc to all nodes in layer $2$, while all nodes in layer $h-1$ are connected to the destination node. All arc costs are generated according to the discrete uniform distribution in the interval $[0,100|j-i|]$. In our experiments we set $h = 10$, $r_{min} = 4$, $r_{max} = 6$ and $p = 0.5$.
\item
\textit{Uniform graphs.} For each two nodes $i$ and $j$ in the graph, directed arc $(i,j)$ exists with probability $p$; see further details on undirected random graphs of this type in~\cite{26}. All arc costs are generated according to the discrete uniform distribution in the interval $[0,100]$. The source and destination nodes are selected as follows. First, given graph $G$ we compute its diameter, $diam(G)$, and construct a list of all node pairs such that the distance between the nodes is approximately $diam(G)/2$. Then a pair of nodes is chosen from this list uniformly at random to be the source and destination nodes. In our experiments we set $|N| = 50$ and $p = 0.5$.


\item
\textit{BA graphs.} These graphs are constructed based on the following preferential attachment mechanism. Suppose $m_0$ is a number of nodes in the initial complete graph. Then in each iteration we add a node and connect it with $m \leq m_0$ existing nodes randomly with probabilities proportional to their degrees; see further details in~\cite{barabasi1999emergence}. We set $|N| = 50$, $m = m_0 = 5$. The arc costs as well as the source and destination nodes are generated in the same manner as for the uniform random graphs described above.
\end{itemize}

\looseness-1\textbf{Initial information for the interdictor.} {In all our experiments we construct $A_0$ as follows. For each arc $a \in A$ we modify its cost to $c_a+M$ with probability $0.5$, where $M=10^4$ is sufficiently large. Thus, the arcs with modified costs do not belong to the shortest path from $s$ to $f$ in the modified graph as long as there exists at least one path that consists of only non-modified arcs. Then we add all arcs in the resulting shortest path to $A_0$. This procedure is repeated $5$ times and hence, $A_0$ consists of arcs that belong to at most $5$ distinct paths from $s$ to $f$.} 

\textbf{Feedback types.} As outlined above we consider two types of information feedback from the evader's action to the interdictor, namely, perfect and noisy feedback. The perfect feedback satisfies Assumption \textbf{A5}. The noisy information feedback is generated as follows. Recall that for each arc $a \in A_0$ the interdictor is assumed to know its nominal cost. Then whenever arc $a \in A \setminus A_0$ is traversed by the evader for the first time instead of $c_a$ the interdictor observes a cost generated according to a uniform distribution from $[c_a-0.2c_a,c_a+0.2c_a]$.

\textbf{Computational settings.} All experiments were performed on a PC with CPU i5-7200U and RAM 6 GB. We use the solution approach from \cite{5} to solve the $k$-most vital arcs problem. The heuristic for the strategic evader (Algorithm \ref{alg1}) is implemented in Java with CPLEX 12.7.1. Furthermore, we set $\alpha=0.5$ and $q=2$ as the heuristic parameters. The intuition behind these settings can be justified as follows.

\looseness-1 Having $\alpha=0.5$ implies that the value of $\delta_\tau$ in decision epoch $\tau$, see line 10 of Algorithm 1, is equal to the average loss of the greedy evader over two consecutive decision epochs, namely, $\tau$ and $\tau+1$. Thus, for the strategic evader, whose goal is to outperform the greedy evader over these two decision epochs, we consider as candidate paths in epoch $\tau$, see line 14, only paths of sufficiently small cost (less than threshold $\delta_\tau$). Following our discussion in Section \ref{sec:algorithm}, if $\alpha < 0.5$, then the algorithm may fail to take into consideration some alternative reasonably good evasion solutions. On the other hand, if $\alpha > 0.5$, then in epoch $\tau$ the heuristic is forced to consider evasion paths with sufficiently large costs. It is reasonable to expect that in most cases these candidate evasion paths in epoch $\tau$ lead to inferior evasion solutions over epochs $\tau$ and $\tau+1$.

Finally, as discussed in Section~\ref{sec:algorithm} the value of $q$ allows the heuristic to control indirectly that the evasion decisions by the strategic evader are sufficiently different from those generated by the greedy evader. From our preliminary experiments not reported here, where we explored sensitivity of Algorithm~\ref{alg1} with respect to the parameter settings, we conclude that Algorithm~\ref{alg1} with $\alpha=0.5$ and $q = 2$ performs reasonably well for our test instances and provides consistent computational results.


\subsection{Results and Discussion}\label{subsec:results}

\textbf{Measures of performance.}  Henceforth, we fix a particular class of random graphs as well as $k \in \{1,\ldots,10\}$ and $T \in \{2,5,10\}$. Then in each \textit{experiment} $Q = 50$ test instances, i.e., random graphs of the specified type, are generated.  We denote by $\chi_{=}(T, k)$ the percentage of test instances  in a particular experiment, where the performance of the greedy policy and the proposed heuristic approach coincide, i.e., the evader's losses obtained by these methods are the same. Also, denote by $\chi_{<}(T, k)$ the percentage of test instances, where the heuristic outperforms the greedy policy. Note that the percentage of test instances, where the heuristic approach is outperformed by the greedy policy is given by:
\begin{equation} \label{percent}
\chi_{>}(T,k) = 100 - \chi_{<}(T,k) - \chi_{=}(T,k).\nonumber
\end{equation}
We use $\chi_{=}$, $\chi_{<}$ and $\chi_{>}$ as the performance measures of the proposed heuristic in Tables~\ref{tab1}-\ref{tab6}.

 Specifically, for each class of random graphs described in Section~\ref{subsec:prelim}, $k \in \{1,\ldots,10\}$ and $T \in \{2,5,10\}$, the experiment is repeated $10$ times. Then in Tables~\ref{tab1}, \ref{tab3} and \ref{tab5} we report average values and standard deviations with respect to $\chi_{<}(T,k)$, $\chi_{=}(T,k)$ and $\chi_{>}(T,k)$.  In Tables~\ref{tab2}, \ref{tab4} and \ref{tab6} the same computational results are provided assuming that the feedback is noisy. Furthermore, in Table \ref{tab7} we report the average running times of the heuristic algorithm for each $k \in \{1,\ldots,10\}$ and $T = 10$.

\begin{table}
\begin{scriptsize}
\begin{center}
\begin{tabular}{|c |ccc |ccc |ccc |}
\hline
& \multicolumn{3}{c|}{$T = 2$} & \multicolumn{3}{c|}{$T =5$} & \multicolumn{3}{c|}{$T =10$} \\ \cline{1-10}
$k$ & $\chi_{<}$ & $\chi_{=}$ & $\chi_{>}$ & $\chi_{<}$ & $\chi_{=}$ & $\chi_{>}$ & $\chi_{<}$ & $\chi_{=}$ & $\chi_{>}$\\ \hline
1 & 4.0 (3.3) & 96.0 (3.3) & 0.0 (0.0) & 5.2 (3.8) & 94.4 (4.0) & 0.4 (0.8) & 5.2 (3.8) & 94.4 (4.0) & 0.4 (0.8) \\
2 & 7.0 (5.0) & 93.0 (5.0) & 0.0 (0.0) & 9.2 (3.6) & 87.6 (6.1) & 3.2 (2.9) & 9.8 (4.1) & 86.0 (6.4) & 4.2 (3.2) \\
3 & 9.2 (4.1) & 90.8 (4.1) & 0.0 (0.0) & 18.2 (4.1) & 76.4 (3.8) & 5.4 (3.2) & 24.2 (8.8) & 69.6 (7.7) & 6.2 (3.4) \\
4 & 10.6 (6.2) & 89.4 (6.2) & 0.0 (0.0) & 23.0 (5.2) & 68.8 (7.0) & 8.2 (5.0) & 31.4 (6.3) & 53.2 (6.8) & 15.4 (6.8) \\ 5 & 14.2 (5.2) & 85.8 (5.2) & 0.0 (0.0) & 32.2 (8.8) & 57.8 (9.8) & 10.0 (2.0) & 44.2 (4.3) & 34.0 (5.7) & 21.8 (2.6) \\
6 & 15.0 (3.6) & 85.0 (3.6) & 0.0 (0.0) & 37.6 (7.1) & 50.8 (8.1) & 11.6 (3.9) & 45.0 (6.7) & 28.0 (5.8) & 27.0 (7.1) \\ 7 & 15.0 (5.8) & 85.0 (5.8) & 0.0 (0.0) & 35.4 (6.3) & 53.7 (5.5) & 10.9 (6.0) & 50.7 (5.2) & 26.5 (7.8) & 22.8 (6.1) \\ 8 & 14.2 (3.9) & 85.8 (3.9) & 0.0 (0.0) & 36.6 (8.0) & 50.4 (7.5) & 13.0 (6.3) & 48.6 (4.4) & 24.8 (6.2) & 26.6 (6.2) \\ 9 & 16.0 (4.0) & 84.0 (4.0) & 0.0 (0.0) & 42.4 (6.2) & 47.4 (5.7) & 10.2 (4.9) & 57.1 (5.0) & 17.3 (3.7) & 25.5 (6.2) \\ 10 & 20.0 (3.1) & 80.0 (3.1) & 0.0 (0.0) & 46.4 (6.2) & 42.8 (4.7) & 10.8 (3.6) & 61.8 (5.4) & 14.3 (4.7) & 23.9 (5.6) \\\hline
\end{tabular}
\caption{\footnotesize Comparison of the greedy policy against the heuristic for the strategic evader in the layered random graphs. The interdictor's feedback is transparent.  For each pair of values of $k$ and $T$, we report the averages and standard deviations of $\chi_{<}(T,k)$, $\chi_{=}(T,k)$ and $\chi_{>}(T,k)$ over 10 experiments. 
}
\label{tab1}
\end{center}
\end{scriptsize}
\end{table}

\begin{table}
\begin{scriptsize}
\begin{center}
\begin{tabular}{|c |ccc |ccc |ccc |}
\hline
& \multicolumn{3}{c|}{$T = 2$} & \multicolumn{3}{c|}{$T =5$} & \multicolumn{3}{c|}{$T =10$} \\ \cline{1-10}
$k$ & $\chi_{<}$ & $\chi_{=}$ & $\chi_{>}$ & $\chi_{<}$ & $\chi_{=}$ & $\chi_{>}$ & $\chi_{<}$ & $\chi_{=}$ & $\chi_{>}$\\ \hline
1 & 4.0 (2.2) & 96.0 (2.2) & 0.0 (0.0) & 10.0 (3.2) & 88.6 (3.0) & 1.4 (0.9) & 9.0 (3.6) & 88.6 (3.0) & 2.4 (1.5) \\
2 & 10.2 (3.3) & 89.8 (3.3) & 0.0 (0.0) & 21.4 (5.7) & 72.2 (7.2) & 6.4 (4.1) & 26.2 (5.4) & 67.4 (6.1) & 6.4 (3.4) \\
3 & 10.4 (2.8) & 89.6 (2.8) & 0.0 (0.0) & 27.2 (5.5) & 66.8 (6.1) & 6.0 (3.0) & 37.6 (5.5) & 51.4 (5.1) & 11.0 (3.9) \\
4 & 10.2 (4.7) & 89.8 (4.7) & 0.0 (0.0) & 30.6 (5.4) & 60.8 (6.9) & 8.6 (4.2) & 49.6 (5.1) & 36.2 (7.0) & 14.2 (4.4) \\
5 & 13.2 (4.5) & 86.8 (4.5) & 0.0 (0.0) & 32.5 (7.7) & 58.1 (7.7) & 9.4 (3.0) & 48.1 (7.9) & 28.8 (4.8) & 23.1 (6.0) \\
6 & 17.0 (5.0) & 83.0 (5.0) & 0.0 (0.0) & 37.2 (5.0) & 52.2 (5.5) & 10.6 (6.0) & 53.9 (4.6) & 21.2 (6.5) & 24.9 (4.4) \\
7 & 16.8 (5.2) & 83.2 (5.2) & 0.0 (0.0) & 42.8 (5.7) & 47.0 (7.4) & 10.2 (3.0) & 58.4 (4.8) & 19.5 (5.8) & 22.0 (3.4) \\
8 & 18.2 (6.8) & 81.8 (6.8) & 0.0 (0.0) & 42.4 (6.5) & 48.8 (8.8) & 8.8 (3.0) & 57.9 (7.4) & 19.2 (6.7) & 22.9 (4.8) \\
9 & 15.2 (3.7) & 84.8 (3.7) & 0.0 (0.0) & 42.2 (5.6) & 45.8 (6.3) & 12.0 (4.7) & 61.4 (5.0) & 16.6 (6.1) & 22.0 (6.0) \\
10 & 17.0 (4.6) & 83.0 (4.6) & 0.0 (0.0) & 44.8 (5.9) & 43.2 (5.5) & 12.0 (4.9) & 59.5 (4.2) & 16.3 (5.1) & 24.3 (6.1) \\
\hline
\end{tabular}
\caption{\footnotesize {Comparison of the greedy policy against the heuristic for the strategic evader in the layered random graphs. The interdictor's feedback is noisy.  For each pair of values of $k$ and $T$, we report the averages and standard deviations of $\chi_{<}(T,k)$, $\chi_{=}(T,k)$ and $\chi_{>}(T,k)$ over 10 experiments. }
}
\label{tab2}
\end{center}
\end{scriptsize}
\end{table}

Finally, in the remainder of the section, whenever we refer to the strategic evader or the heuristic we imply that the evader uses Algorithm~\ref{alg1} in the iterative manner over the specified time horizon $T$ (recall that Algorithm 1 computes an evasion decision only for at most two decision epochs).

\textbf{Heuristic performance vs.~Greedy policy.} First, consider the case of $T = 2$ that corresponds to the first column in each of the tables. We observe that $\chi_{>}(2,k)=0$ for all $k \in \{1,\ldots,10\}$ in all tables. In other words, the strategic evader is always at least as good as the greedy evader for $T=2$. This heuristic performance is expected given that, first, Algorithm \ref{alg1} compares possible alternative solutions with a solution of the greedy policy for two steps ahead; furthermore, whenever the heuristic cannot find a good quality alternative solution, then the heuristic resorts to the greedy policy solution.

\begin{table}
\begin{scriptsize}
\begin{center}
\begin{tabular}{|c |ccc |ccc |ccc |}
\hline
& \multicolumn{3}{c|}{$T = 2$} & \multicolumn{3}{c|}{$T =5$} & \multicolumn{3}{c|}{$T =10$} \\ \cline{1-10}
$k$ & $\chi_{<}$ & $\chi_{=}$ & $\chi_{>}$ & $\chi_{<}$ & $\chi_{=}$ & $\chi_{>}$ & $\chi_{<}$ & $\chi_{=}$ & $\chi_{>}$\\ \hline
1 & 0.8 (1.3) & 99.2 (1.3) & 0.0 (0.0) & 1.2 (1.3) & 98.8 (1.3) & 0.0 (0.0) & 1.2 (1.3) & 98.8 (1.3) & 0.0 (0.0) \\
2 & 4.8 (2.2) & 95.2 (2.2) & 0.0 (0.0) & 6.6 (3.5) & 92.0 (2.4) & 1.4 (2.4) & 6.8 (3.7) & 91.8 (2.6) & 1.4 (2.4) \\
3 & 4.2 (2.6) & 95.8 (2.6) & 0.0 (0.0) & 10.0 (4.5) & 88.0 (5.1) & 2.0 (1.8) & 12.6 (6.0) & 84.6 (6.1) & 2.8 (2.2) \\
4 & 6.2 (1.9) & 93.8 (1.9) & 0.0 (0.0) & 16.6 (6.1) & 78.0 (7.4) & 5.4 (2.5) & 20.2 (4.9) & 71.2 (6.3) & 8.6 (3.7) \\
5 & 6.6 (3.8) & 93.4 (3.8) & 0.0 (0.0) & 16.8 (6.7) & 76.2 (5.5) & 7.0 (3.3) & 24.8 (5.6) & 61.4 (5.3) & 13.8 (2.3) \\
6 & 7.2 (4.2) & 92.8 (4.2) & 0.0 (0.0) & 21.2 (6.7) & 72.4 (7.8) & 6.4 (4.6) & 32.6 (4.9) & 50.2 (5.1) & 17.2 (6.8) \\
7 & 7.0 (2.7) & 93.0 (2.7) & 0.0 (0.0) & 23.2 (6.5) & 70.0 (7.0) & 6.8 (3.2) & 33.8 (6.7) & 46.2 (8.5) & 20.0 (4.3) \\
8 & 7.0 (4.4) & 93.0 (4.4) & 0.0 (0.0) & 25.6 (5.6) & 67.2 (4.9) & 7.2 (3.6) & 37.0 (6.1) & 42.0 (6.4) & 21.0 (3.0) \\
9 & 8.2 (4.3) & 91.8 (4.3) & 0.0 (0.0) & 27.6 (7.0) & 62.8 (5.3) & 9.6 (4.5) & 41.2 (3.6) & 40.0 (5.7) & 18.8 (5.7) \\
10 & 10.8 (3.9) & 89.2 (3.9) & 0.0 (0.0) & 25.8 (8.4) & 67.4 (5.9) & 6.8 (5.0) & 41.6 (6.4) & 40.2 (5.7) & 18.2 (3.6) \\
\hline
\end{tabular}
\caption{\footnotesize Comparison of the greedy policy against the heuristic for the strategic evader in the uniform random graphs. The interdictor's feedback is transparent.  For each pair of values of $k$ and $T$,  we report the averages and standard deviations of $\chi_{<}(T,k)$, $\chi_{=}(T,k)$ and $\chi_{>}(T,k)$ over 10 experiments. 
}
\label{tab3}
\end{center}
\end{scriptsize}
\end{table}

\begin{table}
\begin{scriptsize}
\begin{center}
\begin{tabular}{|c |ccc |ccc |ccc |}
\hline
& \multicolumn{3}{c|}{$T = 2$} & \multicolumn{3}{c|}{$T =5$} & \multicolumn{3}{c|}{$T =10$} \\ \cline{1-10}
$k$ & $\chi_{<}$ & $\chi_{=}$ & $\chi_{>}$ & $\chi_{<}$ & $\chi_{=}$ & $\chi_{>}$ & $\chi_{<}$ & $\chi_{=}$ & $\chi_{>}$\\ \hline
1 & 3.0 (2.0) & 97.0 (2.0) & 0.0 (0.0) & 4.2 (2.6) & 95.8 (2.6) & 0.0 (0.0) & 4.2 (2.6) & 95.8 (2.6) & 0.0 (0.0) \\
2 & 3.8 (2.3) & 96.2 (2.3) & 0.0 (0.0) & 8.2 (1.4) & 89.8 (1.9) & 2.0 (1.8) & 9.2 (1.3) & 89.0 (2.0) & 1.8 (1.7) \\
3 & 7.0 (3.5) & 93.0 (3.5) & 0.0 (0.0) & 18.2 (4.2) & 78.0 (3.3) & 3.8 (2.9) & 20.2 (4.2) & 72.8 (3.9) & 7.0 (3.3) \\
4 & 7.0 (3.3) & 93.0 (3.3) & 0.0 (0.0) & 19.4 (5.4) & 75.8 (6.0) & 4.8 (2.7) & 28.8 (3.0) & 59.2 (5.3) & 12.0 (4.5) \\
5 & 8.2 (3.7) & 91.8 (3.7) & 0.0 (0.0) & 22.4 (4.3) & 72.6 (5.6) & 5.0 (3.9) & 33.6 (5.4) & 51.8 (5.0) & 14.6 (4.6) \\
6 & 9.0 (2.9) & 91.0 (2.9) & 0.0 (0.0) & 22.8 (2.6) & 69.2 (4.4) & 8.0 (4.0) & 36.4 (6.8) & 45.6 (7.5) & 18.0 (5.0) \\
7 & 7.4 (4.6) & 92.6 (4.6) & 0.0 (0.0) & 21.8 (4.2) & 71.0 (5.0) & 7.2 (2.4) & 33.4 (3.5) & 46.4 (6.2) & 20.2 (6.1) \\
8 & 11.2 (4.7) & 88.8 (4.7) & 0.0 (0.0) & 27.0 (5.6) & 65.4 (5.9) & 7.6 (2.2) & 42.2 (5.2) & 38.2 (4.2) & 19.6 (5.9) \\
9 & 10.6 (4.4) & 89.4 (4.4) & 0.0 (0.0) & 28.0 (6.2) & 64.6 (4.4) & 7.4 (3.9) & 45.6 (5.4) & 35.8 (7.0) & 18.6 (3.0) \\
10 & 11.4 (2.7) & 88.6 (2.7) & 0.0 (0.0) & 27.8 (4.4) & 63.2 (4.7) & 9.0 (2.2) & 41.8 (10.5) & 38.2 (10.0) & 20.0 (5.9) \\

\hline
\end{tabular}
\caption{\footnotesize {Comparison of the greedy policy against the heuristic for the strategic evader in the uniform random graphs. The interdictor's feedback is noisy.  For each pair of values of $k$ and $T$,  we report the averages and standard deviations of $\chi_{<}(T,k)$, $\chi_{=}(T,k)$ and $\chi_{>}(T,k)$ over 10 experiments. }
}
\label{tab4}
\end{center}
\end{scriptsize}
\end{table}

\looseness-1{Next, it is worth highlighting that the values of both $\chi_{>}(T,1)$ and $\chi_{<}(T,1)$ are rather small for all test instances, i.e., the strategic evader almost always selects a greedy evasion solution when the interdictor blocks exactly one arc. This observation can be explained by Proposition~\ref{corrollary2}, where we show that for $T=2$, $A_0=\emptyset$ and $k=1$ the greedy policy is optimal (recall that in our test instances $A_0\neq \emptyset$). Thus, we can expect that greedy evasion solutions are often close to optimal for $k=1$ whenever $|A_0|$ is sufficiently small. Also, the aforementioned trend is less pronounced for instances with noisy feedback (if one compares the results for $\chi_{>}(T,1)$ in Tables~\ref{tab1}, \ref{tab3} and \ref{tab5} against those in Tables~\ref{tab2}, \ref{tab4} and \ref{tab6} , respectively), which is also reasonable to expect.}

For a fixed value of $k$, the value of $\chi_{<}(T,k)$ tends to grow with the increase in $T$. In other words, the number of instances, where the heuristic outperforms the greedy policy increases. This observation is intuitive given that as the number of decision epochs increases there are more opportunities for the strategic evader to improve her performance.

On the other hand, we observe that the value of $\chi_{>}(T,k)$ also increases with $T$. For example, for $k = 4$ in Table \ref{tab1}, $\chi_{>}(T,4)$ monotonously increases from $0$ to $15.4$. It implies that as $T$ increases there exist instances, where the greedy policy outperforms the proposed heuristic. This fact is also not surprising if one recalls that the heuristic is myopic as it computes evasion decisions only for at most two decision epochs by applying Algorithm~\ref{alg1} in the iterative manner. Nevertheless, in all cases we have that $\chi_{<}(T,k) > \chi_{>}(T,k)$ on average, which implies that the evader should prefer using the heuristic than the greedy policy. Note that $\chi_{=}(T, k)$~(i.e., the number of instances where the performances of the approaches coincide) decreases as $T$ increases, which is natural given the above discussion.

\begin{table}
\begin{scriptsize}
\begin{center}
\begin{tabular}{|c |ccc |ccc |ccc |}
\hline
& \multicolumn{3}{c|}{$T = 2$} & \multicolumn{3}{c|}{$T =5$} & \multicolumn{3}{c|}{$T =10$} \\ \cline{1-10}
$k$ & $\chi_{<}$ & $\chi_{=}$ & $\chi_{>}$ & $\chi_{<}$ & $\chi_{=}$ & $\chi_{>}$ & $\chi_{<}$ & $\chi_{=}$ & $\chi_{>}$\\ \hline
1 & 2.4 (2.2) & 97.6 (2.2) & 0.0 (0.0) & 3.2 (2.4) & 96.2 (2.7) & 0.6 (1.8) & 3.2 (2.4) & 96.2 (2.7) & 0.6 (1.8) \\
2 & 6.2 (3.8) & 93.8 (3.8) & 0.0 (0.0) & 13.0 (6.9) & 85.0 (7.1) & 2.0 (2.5) & 13.0 (6.8) & 84.8 (6.8) & 2.2 (2.6) \\
3 & 9.4 (4.0) & 90.6 (4.0) & 0.0 (0.0) & 21.2 (6.1) & 73.8 (5.3) & 5.0 (3.1) & 26.2 (6.5) & 67.8 (7.3) & 6.0 (2.5) \\
4 & 15.0 (4.2) & 85.0 (4.2) & 0.0 (0.0) & 29.2 (6.7) & 59.6 (6.4) & 11.2 (4.2) & 38.4 (5.6) & 43.0 (6.6) & 18.6 (5.3) \\
5 & 17.2 (3.2) & 82.8 (3.2) & 0.0 (0.0) & 34.6 (7.6) & 53.9 (7.7) & 11.6 (4.1) & 40.0 (12.0) & 37.6 (13.6) & 22.4 (10.6) \\
6 & 16.2 (5.2) & 83.8 (5.2) & 0.0 (0.0) & 44.1 (6.0) & 47.0 (6.0) & 8.9 (3.0) & 43.1 (7.7) & 29.7 (9.5) & 27.2 (10.4) \\
7 & 14.8 (4.2) & 85.2 (4.2) & 0.0 (0.0) & 37.5 (8.1) & 48.5 (7.4) & 14.1 (5.0) & 52.3 (10.6) & 24.6 (7.7) & 23.1 (10.8) \\
8 & 17.2 (3.5) & 82.8 (3.5) & 0.0 (0.0) & 47.5 (6.6) & 40.3 (5.8) & 12.2 (4.9) & 56.0 (9.1) & 23.4 (8.1) & 20.7 (9.8) \\
9 & 19.2 (4.9) & 80.8 (4.9) & 0.0 (0.0) & 44.4 (6.2) & 41.1 (4.2) & 14.6 (3.9) & 52.1 (14.4) & 22.8 (7.2) & 25.1 (12.5) \\
10 & 15.8 (4.3) & 84.2 (4.3) & 0.0 (0.0) & 39.7 (7.1) & 47.3 (6.7) & 13.0 (2.4) & 57.8 (19.4) & 19.7 (14.0) & 22.4 (8.9) \\
\hline
\end{tabular}
\caption{\footnotesize Comparison of the greedy policy against the heuristic for the strategic evader in the BA random graphs. The interdictor's feedback is transparent.  For each pair of values of $k$ and $T$,  we report the averages and standard deviations of $\chi_{<}(T,k)$, $\chi_{=}(T,k)$ and $\chi_{>}(T,k)$ over 10 experiments. 
}
\label{tab5}
\end{center}
\end{scriptsize}
\end{table}

\begin{table}
\begin{scriptsize}
\begin{center}
\begin{tabular}{|c |ccc |ccc |ccc |}
\hline
& \multicolumn{3}{c|}{$T = 2$} & \multicolumn{3}{c|}{$T =5$} & \multicolumn{3}{c|}{$T =10$} \\ \cline{1-10}
$k$ & $\chi_{<}$ & $\chi_{=}$ & $\chi_{>}$ & $\chi_{<}$ & $\chi_{=}$ & $\chi_{>}$ & $\chi_{<}$ & $\chi_{=}$ & $\chi_{>}$\\ \hline
1 & 3.2 (2.2) & 96.8 (2.2) & 0.0 (0.0) & 5.4 (2.7) & 94.0 (2.7) & 0.6 (0.9) & 5.4 (2.7) & 94.0 (2.7) & 0.6 (0.9) \\
2 & 7.0 (4.2) & 93.0 (4.2) & 0.0 (0.0) & 17.8 (6.3) & 79.2 (6.1) & 3.0 (1.8) & 17.6 (4.9) & 77.8 (6.1) & 4.6 (2.7) \\
3 & 9.2 (4.2) & 90.8 (4.2) & 0.0 (0.0) & 28.0 (4.1) & 66.8 (6.8) & 5.2 (4.8) & 35.8 (5.2) & 55.0 (7.4) & 9.2 (3.9) \\
4 & 12.6 (3.4) & 87.4 (3.4) & 0.0 (0.0) & 31.0 (10.2) & 59.8 (11.6) & 9.2 (5.1) & 46.8 (8.4) & 37.4 (10.3) & 15.8 (3.9) \\
5 & 16.2 (4.9) & 83.8 (4.9) & 0.0 (0.0) & 40.3 (4.8) & 50.5 (5.5) & 9.2 (2.7) & 48.2 (6.7) & 30.3 (5.9) & 21.5 (8.7) \\
6 & 16.6 (5.0) & 83.4 (5.0) & 0.0 (0.0) & 38.7 (8.9) & 49.2 (8.0) & 12.1 (4.4) & 50.6 (6.9) & 24.3 (8.2) & 25.1 (7.4) \\
7 & 16.2 (5.1) & 83.8 (5.1) & 0.0 (0.0) & 43.8 (3.2) & 45.2 (7.2) & 11.0 (4.5) & 56.0 (12.8) & 17.9 (7.0) & 26.1 (9.2) \\
8 & 18.2 (5.0) & 81.8 (5.0) & 0.0 (0.0) & 44.4 (7.0) & 43.9 (6.1) & 11.7 (5.0) & 52.1 (7.8) & 17.6 (7.9) & 30.3 (8.3) \\
9 & 20.8 (6.7) & 79.2 (6.7) & 0.0 (0.0) & 46.9 (5.7) & 39.8 (6.0) & 13.3 (2.3) & 52.3 (15.2) & 17.0 (8.3) & 30.7 (10.8) \\
10 & 17.6 (6.4) & 82.4 (6.4) & 0.0 (0.0) & 43.4 (6.6) & 44.3 (7.6) & 12.3 (4.4) & 54.4 (16.9) & 17.6 (13.0) & 28.0 (13.7) \\
\hline
\end{tabular}
\caption{\footnotesize{Comparison of the greedy policy against the heuristic for the strategic evader in the BA random graphs. The interdictor's feedback is noisy.  For each pair of values of $k$ and $T$, we report the averages and standard deviations of $\chi_{<}(T,k)$, $\chi_{=}(T,k)$ and $\chi_{>}(T,k)$ over 10 experiments. }}
\label{tab6}
\end{center}
\end{scriptsize}
\end{table}

Furthermore, observe that for fixed $T$ with the increase of $k$, the value of $\chi_{<}(T,k)$ tends to grow, while $\chi_{=}(T,k)$ tends to decrease. The latter observation can be explained as follows.  {As we point out in Section \ref{sec:theory}, if parameter $k$ is sufficiently large, then the optimal evasion solution is likely to be arc-disjoint (recall Remark~\ref{corrollary3}). However, an optimal pair of arc-disjoint paths usually does not coincide with the greedy evasion solution; recall Examples \ref{Example 1} and \ref{Example 2}.}

{Finally, when comparing the results Tables~\ref{tab1}, \ref{tab3} and \ref{tab5} against those in Tables~\ref{tab2}, \ref{tab4} and \ref{tab6} , respectively, we conclude that the above observations hold for both perfect and noisy feedback settings (recall that the latter relaxes Assumption~\textbf{A5} used in our derivations of the theoretical results in Section~\ref{sec:NP} and~\ref{sec:theory}). In other words, in scenarios where the interdictor does not have the perfect information about the evader's actions the heuristic performs sufficiently well.}


\begin{table}[t]
\begin{scriptsize}
\begin{center}
\begin{tabular}{|c |c |c | c|}
\hline
$k$ & Layered graphs & Uniform graphs & BA graphs \\
\hline
1 & 0.06 & 0.08 & 0.06 \\
2 & 0.2 & 0.21 & 0.2 \\
3 & 0.5 & 0.39 & 0.42 \\
4 & 0.97 & 0.61 & 0.72 \\
5 & 1.49 & 0.84 & 0.9 \\
6 & 1.72 & 0.81 & 0.87 \\
7 & 1.93 & 0.74 & 0.88 \\
8 & 2.07 & 0.72 & 0.82 \\
9 & 2.02 & 0.7 & 0.62 \\
10 & 1.52 & 0.67 & 0.53 \\
\hline
\end{tabular}
\caption{\footnotesize {The average running times (in seconds) of the heuristic algorithm for the strategic evader (Algorithm~\ref{alg1}) applied for $T = 10$ and $k \in \{1,\ldots,10\}$.}
}
\label{tab7}
\end{center}
\end{scriptsize}
\end{table}

\textbf{Running time.} Note that the running time of the heuristic algorithm initially increases and then tends to decrease with the increase of parameter $k$, see Table \ref{tab7}. Clearly, if $k$ is sufficiently small, then the $k$-most vital arcs problem can be solved fast and thus, the total running time of Algorithm \ref{alg1} decreases. Alternatively, the larger the value of parameter $k$ is, the more information the evader is forced to reveal to the interdictor in each decision epoch. Once the interdictor finds an optimal solution of the $k$-most vital arcs problem in the full information network, then the interdictor keeps implementing the same solution in each decision epoch, see~\cite{1} for more details. Thus, as the heuristic is myopic, then it also does not change the evasion solutions in the corresponding decision epochs. Hence, for sufficiently large values of parameter~$k$ most of the computational efforts of the heuristic algorithm are restricted to the first few decision epochs.

\textbf{Summary.} To conclude, the proposed heuristic approach for the strategic evader outperforms in general the greedy evasion policy. Moreover, despite its myopic two-step look ahead structure our approach can be successfully applied even for rather large values of parameter $T$ and any set of initial information available to the interdictor. Finally, we note that the obtained results are pretty consistent with respect to all considered classes of random graphs and types of the information feedback obtained by the interdictor from the evader's actions.

\section{Conclusions} \label{sec:conclusions}

\looseness-1In this study, we consider a class of sequential interdiction settings where the interdictor has incomplete initial information about the network while the evader has complete knowledge of the network including its structure and arc costs. However, by observing the evader's actions, the interdictor learns about the network structure and costs and thus, can adjust his actions in subsequent decision epochs.

\looseness-1 Our focus is on the evader's perspective. In particular, we assume that the interdictor acts in a greedy manner, by blocking $k$ arcs known to him in each round. Such interdiction policies are known from the literature, see~\cite{1,2}, to perform well against greedy evaders who traverse along the shortest path in the non-interdicted network in each round. In this paper, our goal is to explore whether the evader can improve her performance by acting in a strategic manner.

Specifically, we first show that the evader's problem is computationally hard even for two decision epochs. Then we derive basic constructive properties of optimal evasion policies for two decision epochs when the interdictor has no initial information about the network structure. \looseness-1Furthermore, based on these observations, we design a heuristic algorithm for a strategic evader in a general setting with an arbitrary time horizon and any initial information available to the interdictor. Our computational experiments demonstrate that the proposed heuristic consistently outperforms the greedy evasion policy on several classes of synthetic network instances with respect to various sets of initial information and types of the feedback available to the interdictor.

The insights (both from theoretical and practical perspectives) from our results can be summarized as follows. First, our heuristic algorithm is based on a two-step look-ahead idea. Thus, whenever the interdictor is myopic the evaders can significantly improve their performance over multiple decision epochs in a rather simple manner, i.e., by considering only two consecutive decision epochs instead of one in the greedy~(i.e., myopic) evasion approach. Furthermore, for the aforementioned two consecutive epochs it may be sufficient for the evader to consider two simple decision strategies. Either the evader should follow a myopic shortest-path based policy or seek only two alternative paths that have some arcs in common with greedy evasion paths in the network. Also, for the latter case the evasion paths in consequent decision epochs often need to be arc-disjoint, which is also rather intuitive from the real-life perspective (e.g., in ``infiltration'' or ``smuggling'' scenarios).

\looseness-1Another interesting observation (with possible practical implications) from our computational study is that strategic evasion decisions are more beneficial when the interdictor has more resources. In other words, the more arcs can be blocked by the interdictor in each decision epoch the faster he forces the greedy evader to reveal new network information, which subsequently can be exploited to improve the interdiction decisions. Thus, in such scenarios the evaders need to follow more sophisticated policies in order to improve their performance. Moreover, this observation also suggests that the availability of more interdiction resources do not necessarily imply that the interdictor may simply rely on myopic policies. From the practical perspective, the increase in the interdiction resources may force the evader to switch to strategic evasions, which in turn decreases the efficiency of myopic interdiction policies and consequently may require from the interdictor implementation of non-myopic decisions. On the other hand, for instances where the interdictor's resources are limited greedy evasions can be close to optimal, which implies that the evaders may simplify their decisions over multiple epochs. Then the interdictor may also consider relying on myopic policies as they perform well against greedy evaders.

With respect to future research directions, it would be interesting to study a setting, where the decision-makers operate in a stochastic environment, see, e.g., related works in \cite{16, 22} for the analysis of different types of feedback in a classical multi-armed bandit setting. Moreover, in our study we assume that the evader observes the blocking decision before choosing a path and knows the interdictor's budget. Therefore, one may explore the evader's problem with this assumption relaxed. For example, the study in \cite{sefair2016dynamic} considers a class of interdiction problems where at every node in the evader's path the interdictor must decide whether or not to expand the interdiction set.
%

\looseness-1 Finally, in this paper our major goal is to gain insights about the evasion policies that could be effective against interdictors who assume simplistic greedy evasion policies. In view of the discussion above, we envision that the results and observations of this study can be further exploited for the development of advanced interdiction models, where the evader may follow more sophisticated evasion policies than those that are typically assumed in the related literature.

\textbf{Acknowledgments.} The article was prepared within the framework of the Basic Research Program at the National Research University Higher School of Economics (HSE). The authors are thankful to the associate editor and three anonymous referees for their constructive comments that allowed us to greatly improve the paper.

\section*{References}

\bibliographystyle{abbrv}
\bibliography{bibliography}

\begin{thebibliography}{10}

\bibitem{20}
Y.~Abbasi, P.~L. Bartlett, V.~Kanade, Y.~Seldin, and C.~Szepesv{\'a}ri.
\newblock Online learning in markov decision processes with adversarially
  chosen transition probability distributions.
\newblock In C.~J.~C. Burges, L.~Bottou, M.~Welling, Z.~Ghahramani, and K.~Q.
  Weinberger, editors, {\em Advances in neural information processing systems},
  pages 2508--2516, 2013.

\bibitem{Ahuja}
R.~K. Ahuja, K.~Mehlhorn, J.~Orlin, and R.~E. Tarjan.
\newblock Faster algorithms for the shortest path problem.
\newblock {\em Journal of the ACM (JACM)}, 37(2):213--223, 1990.

\bibitem{allain2016evolving}
R.~J. Allain.
\newblock An evolving asymmetric game for modeling interdictor-smuggler
  problems.
\newblock Technical report, Naval Postgraduate School, 2016.

\bibitem{16}
P.~Auer, N.~Cesa-Bianchi, and P.~Fischer.
\newblock Finite-time analysis of the multiarmed bandit problem.
\newblock {\em Machine Learning}, 47(2-3):235--256, 2002.

\bibitem{barabasi1999emergence}
A.-L. Barab{\'a}si and R.~Albert.
\newblock Emergence of scaling in random networks.
\newblock {\em Science}, 286(5439):509--512, 1999.

\bibitem{1}
J.~S. Borrero, O.~A. Prokopyev, and D.~Saur{\'e}.
\newblock Sequential shortest path interdiction with incomplete information.
\newblock {\em Decision Analysis}, 13(1):68--98, 2015.

\bibitem{2}
J.~S. Borrero, O.~A. Prokopyev, and D.~Saur{\'e}.
\newblock Sequential interdiction with incomplete information and learning.
\newblock {\em Operations Research}, 67(1):72--89, 2019.

\bibitem{brown2009interdicting}
G.~G. Brown, W.~M. Carlyle, R.~C. Harney, E.~M. Skroch, and R.~K. Wood.
\newblock Interdicting a nuclear-weapons project.
\newblock {\em Operations Research}, 57(4):866--877, 2009.

\bibitem{brown2006defending}
G.~G. Brown, W.~M. Carlyle, J.~Salmer{\'o}n, and R.~K. Wood.
\newblock Defending critical infrastructure.
\newblock {\em Interfaces}, 36(6):530--544, 2006.

\bibitem{22}
S.~Bubeck.
\newblock Introduction to online optimization.
\newblock {\em Lecture Notes}, pages 1--86, 2011.

\bibitem{Colson2007}
B.~Colson, P.~Marcotte, and G.~Savard.
\newblock An overview of bilevel optimization.
\newblock {\em Annals of Operations Research}, 153(1):235--256, 2007.

\bibitem{11}
H.~Corley and Y.~S. David.
\newblock Most vital links and nodes in weighted networks.
\newblock {\em Operations Research Letters}, 1(4):157--160, 1982.

\bibitem{Dimitrov2013}
N.~B. Dimitrov and D.~P. Morton.
\newblock Interdiction models and applications.
\newblock In J.~W. Herrmann, editor, {\em Handbook of Operations Research for
  Homeland Security}, pages 73--103. Springer, 2013.

\bibitem{enayaty2018logic}
F.~Enayaty-Ahangar, C.~E. Rainwater, and T.~C. Sharkey.
\newblock A logic-based decomposition approach for multi-period network
  interdiction models.
\newblock {\em Omega}, 2018.
\newblock to appear.

\bibitem{26}
P.~Erdos and A.~R{\'e}nyi.
\newblock On the evolution of random graphs.
\newblock {\em Publ. Math. Inst. Hung. Acad. Sci}, 5(1):17--60, 1960.

\bibitem{GJ}
M.~R. Garey and D.~S. Johnson.
\newblock {\em Computers and Intractability: A Guide to the Theory of
  NP-Completeness}.
\newblock Freeman, New York, 1979.

\bibitem{3}
A.~Gutfraind, A.~Hagberg, and F.~Pan.
\newblock Optimal interdiction of unreactive {M}arkovian evaders.
\newblock In {\em International Conference on AI and OR Techniques in
  Constriant Programming for Combinatorial Optimization Problems}, pages
  102--116. Springer, 2009.

\bibitem{harris1955fundamentals}
T.~Harris and F.~Ross.
\newblock Fundamentals of a method for evaluating rail net capacities.
\newblock Technical report, The RAND Corporation, Santa Monica, California,
  1955.

\bibitem{Held2005}
H.~Held and D.~L. Woodruff.
\newblock Heuristics for multi-stage interdiction of stochastic networks.
\newblock {\em Journal of Heuristics}, 11(5-6):483--500, 2005.

\bibitem{5}
E.~Israeli and R.~K. Wood.
\newblock Shortest-path network interdiction.
\newblock {\em Networks}, 40(2):97--111, 2002.

\bibitem{Janjarassuk2008}
U.~Janjarassuk and J.~Linderoth.
\newblock Reformulation and sampling to solve a stochastic network interdiction
  problem.
\newblock {\em Networks}, 52(3):120--132, 2008.

\bibitem{17}
V.~Kanade and T.~Steinke.
\newblock Learning hurdles for sleeping experts.
\newblock {\em ACM Transactions on Computation Theory (TOCT)}, 6(3):11, 2014.

\bibitem{14}
R.~Kleinberg, A.~Niculescu-Mizil, and Y.~Sharma.
\newblock Regret bounds for sleeping experts and bandits.
\newblock {\em Machine Learning}, 80(2-3):245--272, 2010.

\bibitem{15}
W.~M. Koolen, M.~K. Warmuth, and D.~Adamskiy.
\newblock Open problem: Online sabotaged shortest path.
\newblock In P.~Grünwald, E.~Hazan, and S.~Kale, editors, {\em Conference on
  Learning Theory}, pages 1764--1766, 2015.

\bibitem{18}
C.-L. Li, S.~Thomas~McCormick, and D.~Simchi-Levi.
\newblock Finding disjoint paths with different path-costs: Complexity and
  algorithms.
\newblock {\em Networks}, 22(7):653--667, 1992.

\bibitem{malaviya2012multi}
A.~Malaviya, C.~Rainwater, and T.~Sharkey.
\newblock Multi-period network interdiction problems with applications to
  city-level drug enforcement.
\newblock {\em IIE Transactions}, 44(5):368--380, 2012.

\bibitem{morton2007models}
D.~P. Morton, F.~Pan, and K.~J. Saeger.
\newblock Models for nuclear smuggling interdiction.
\newblock {\em IIE Transactions}, 39(1):3--14, 2007.

\bibitem{Nehme2009}
M.~V. Nehme.
\newblock Two-person games for stochastic network interdiction: models,
  methods, and complexities.
\newblock 2009.

\bibitem{Pan2008}
F.~Pan and D.~P. Morton.
\newblock Minimizing a stochastic maximum-reliability path.
\newblock {\em Networks: An International Journal}, 52(3):111--119, 2008.

\bibitem{Sanjab2017}
A.~Sanjab, W.~Saad, and T.~Ba{\c{s}}ar.
\newblock Prospect theory for enhanced cyber-physical security of drone
  delivery systems: A network interdiction game.
\newblock In {\em 2017 IEEE International Conference on Communications (ICC)},
  pages 1--6. IEEE, 2017.

\bibitem{Sanjab2019}
A.~Sanjab, W.~Saad, and T.~Ba{\c{s}}ar.
\newblock A game of drones: Cyber-physical security of time-critical uav
  applications with cumulative prospect theory perceptions and valuations.
\newblock {\em arXiv preprint arXiv:1902.03506}, 2019.

\bibitem{12}
B.~Schieber, A.~Bar-Noy, and S.~Khuller.
\newblock The complexity of finding most vital arcs and nodes.
\newblock {\em Technical Report CS-TR35-39}, 1995.

\bibitem{sefair2016dynamic}
J.~A. Sefair and J.~C. Smith.
\newblock Dynamic shortest-path interdiction.
\newblock {\em Networks}, 68(4):315--330, 2016.

\bibitem{Sinha2018}
A.~Sinha, F.~Fang, B.~An, C.~Kiekintveld, and M.~Tambe.
\newblock Stackelberg security games: Looking beyond a decade of success.
\newblock In {\em IJCAI}, pages 5494--5501, 2018.

\bibitem{smith2008algorithms}
J.~C. Smith and C.~Lim.
\newblock Algorithms for network interdiction and fortification games.
\newblock In A.~Chinchuluun, P.~M. Pardalos, A.~Migdalas, and L.~Pitsoulis,
  editors, {\em Pareto optimality, game theory and equilibria}, pages 609--644.
  Springer, 2008.

\bibitem{smith2013modern}
J.~C. Smith, M.~Prince, and J.~Geunes.
\newblock Modern network interdiction problems and algorithms.
\newblock In P.~M. Pardalos, D.-Z. Du, and R.~L. Graham, editors, {\em Handbook
  of combinatorial optimization}, pages 1949--1987. Springer, 2013.

\bibitem{Smith2019}
J.~C. Smith and Y.~Song.
\newblock A survey of network interdiction models and algorithms.
\newblock {\em European Journal of Operational Research}, 2019.
\newblock to appear.

\bibitem{sullivan2014exact}
K.~M. Sullivan and J.~C. Smith.
\newblock Exact algorithms for solving a {E}uclidean maximum flow network
  interdiction problem.
\newblock {\em Networks}, 64(2):109--124, 2014.

\bibitem{wood2011bilevel}
R.~K. Wood.
\newblock Bilevel network interdiction models: Formulations and solutions.
\newblock In J.~J. Cochran, L.~A.~C. Jr., P.~Keskinocak, J.~P. Kharoufeh, and
  J.~C. Smith, editors, {\em Wiley Encyclopedia of Operations Research and
  Management Science}, pages 1--11. John Wiley \& Sons, Inc, 2010.

\bibitem{Zheng2012}
J.~Zheng and D.~A. Casta{\~n}{\'o}n.
\newblock Dynamic network interdiction games with imperfect information and
  deception.
\newblock In {\em 2012 IEEE 51st IEEE Conference on Decision and Control
  (CDC)}, pages 7758--7763. IEEE, 2012.

\end{thebibliography}
\end{document}